%% file: paper_20240324.tex
\documentclass[a4paper,11pt,oneside]{article}

\PassOptionsToPackage{medium}{titlesec}

\input{packages/packages_20240324.tex}

\graphicspath{{images/}}

\title{Analytical valuation of vulnerable derivative claims with bilateral cash flows under credit, funding and wrong way risk}
\author{Juan Jos\'e Francisco Migu\'elez%
	    \thanks{This research is part of the PhD thesis in Financial Mathematics of Juan Jos\'e Francisco Migu\'elez at King's College London. The views expressed in this paper are those of the author and do not necessarily represent those of current or past employers.} \\
	\small Department of Mathematics, King's College London\\
	\small CVA \& Counterparty Model Risk, Bank of America
	\\[7pt]
	Cristin Buescu\\
	\small Department of Mathematics, King's College London}

\date{\it 
	This version: \today{}.}

\begin{document}
\maketitle

\begin{abstract}
	\noindent
	We study the problem of valuing and hedging a vulnerable derivative claim with bilateral cash flows between two counterparties in the presence of asymmetric funding costs, defaults and wrong way risk (WWR). We characterize the pre-default claim value as the solution to a non-linear Cauchy problem. We show an explicit stochastic representation of the solution exists under a funding policy which linearises the Cauchy PDE. We apply this framework to the valuation of a vulnerable equity forward and show it can be represented as a portfolio of European options. Despite the complexity of the model, we prove the forward's value admits an analytical formula involving only elementary functions and Gaussian integrals.  Based on this explicit formula, numerical analysis demonstrates WWR has a significant impact even under benign assumptions: with a parameter configuration less punitive than that representative of Archegos AM default, we find WWR can shift values for vulnerable forwards by 100bps of notional, while peak exposures increase by 25\% of notional. This framework is the first to apply to contracts with bilateral cash flows in the presence of credit, funding and WWR, resulting in a non-linear valuation formula which admits a closed-form solution under a suitable funding policy.
\end{abstract}


\section{Introduction}
\label{sec:introduction}

In the last decade and a half, financial markets have undergone significant structural shifts motivated by changes in both attitudes and the regulatory environment. As a consequence, market participants started recognizing additional risk factors and trading costs, when pricing derivative transactions, which were either being neglected or inexistent prior to the crisis. Academic and industry research has sought to incorporate such factors into derivative valuation theory, in an effort which continues to this day.
\newline

Most of these changes can be traced back to the 2007-2008 Global Financial Crisis (GFC) and its ramifications. The sudden default of the US broker-dealer Lehman Brothers in September 2008 exposed the default risk faced by counterparties in any bilateral derivative transaction. From then on pricing models had to incorporate the cost of hedging this bilateral credit risk which came to be known as credit valuation adjustment (CVA), modelled as an additional component to the claim's risk-free value \citep{Gregory2009}. 
\newline

The Lehman Brothers default and the ensuing financial panic re-acquainted participants with liquidity risk in the wholesale lending market. Term premiums (also known as basis) across interest rate tenors, which had been negligible prior to the crisis, widened throughout Autumn 2008 and traders had to adapt their yield curve models to a multi-curve environment \citep{Bianchetti2010}. In cross-currency markets, an additional basis between funding costs in different currencies emerged at the same time and has persisted since then \citep{FujiiShimadaTakahashi2010,BorioMccauleyMcguireSushko2016}. Under this multi-curve environment, traders and investors had to start accounting for different rates for borrowing or lending funds when valuing derivative transactions \citep{Mercurio2014}. Some participants started introducing funding valuation adjustments (FVA) into their pricing models to account for any asymmetrical financing cost arising from these transactions \citep{HullWhite2012}. 
\newline

Changes to traditional credit and funding assumptions were compounded by the response from global regulators to the crisis. The banking regulation architecture was overhauled through the Basel III accords \mycitepalias{}{}{BCBS2011} which emphasized credit risk mitigation for derivative contracts. Regulatory bodies devised incentives which encouraged participants in derivative markets to either execute transactions in central counterparties, or to secure over-the-counter (OTC) deals with sufficient collateral. These new collateralization practices forced the industry to review traditional discounting assumptions made in valuation models \citep{Piterbarg2010}. For non-collateralized transactions, Basel III imposed further capital requirements which had to be recognized by large derivative dealers in the form of capital valuation adjustments (KVA) when valuing and hedging derivatives \citep{GreenKenyonDennis2014}. 
\newline

The simultaneous occurrence of Lehman Brothers default and a stock market crash at the height of the GFC lead to a literature stream on modelling adverse correlation between asset prices and credit events (known as \emph{wrong way risk}, or WWR) as well as jumps in market prices when a market participant defaults (\emph{gap risk})%
	\footnote{Loosely speaking wrong way risk is a pre-default risk while gap risk materialises at the time of default. By abuse of language we might write ``wrong way risk'' when referring to both phenomena.}.
These features can be captured by augmenting asset price dynamics with the default process driving default \citep{MercurioLi2015}. This approach has recently been used for valuing CVA for quanto CDS contracts, where the FX rate usually jumps whenever the sovereign defaults \citep{BrigoPedePetrelli2019}. More recently, the default of Archegos Asset Management has brought wrong way and gap risks to the attention of financial regulators: the large losses incurred by broker-dealers during this episode revealed the risk that exposures might unexpectedly jump by a large amount at default \citep{Anfuso2023,Arnsdorf2023}.
\newline

There are a number of references that provide a comprehensive review of the topic of derivative valuation under counterparty risks. The monographs by \citet{BieleckiRutkowski2004} and \citet{CesariEtAl2009} cover the mathematical tools for valuing financial derivatives under credit risk. From an applied perspective, the reader can refer to the books by \citet{Green2015} and \citet{Gregory2020} for a detailed overview of post-GFC changes to derivation valuation models to include credit and funding risk factors. Finally, the collection edited by \citet{KenyonGreen2016} includes landmark research on topics such as valuations adjustments, collateral modelling and WWR.
\newline

In this article, we present a comprehensive model which incorporates \emph{ab initio} most credit and funding risk factors described above, instead of treating them as independent components to the risk-free valuation. Our model remains sufficiently tractable to provide closed-form formul\ae\ for valuing uncollateralized derivative claims with bilateral cash flows such as forwards -- instead of relying on numerical methods, as has been customary in most of the CVA/FVA literature. Our model extends the analytical framework from \citet{BlackScholes1973} to a market with differential rates, defaults and WWR.
\newline

We generalize the frameworks from \citet{BurgardKjaer2011}, \citet{MercurioLi2015} and \citet{BrigoBuescuRutkowski2017} which posit financial markets with heterogeneous rates, defaultable counterparties, or wrong way risks. In the first paper, the authors derive partial differential equations (PDE) for the valuation of derivative claims subject to bilateral default risk under different recovery assumptions, yet explicit valuation formulas are absent when market prices and credit events are not independent. In the second paper, the authors introduce jumps to model adverse price moves between asset prices and the default event; their numerical analysis shows WWR has a material effect on valuations but they do not include any analytical formula. Finally, the last paper presents a valuation formula for a European call option subject to default risk, with no recovery at default, and in the presence of a funding spread; the authors interpret their result as the classical Black-Scholes formula with a modified dividend yield to account for credit and funding spreads. In the last two papers, the authors only deal with the case where the derivative claim has unilateral cash flows and do not include WWR.
\newline

In contrast, we consider claims with bilateral cash flows traded in a market with credit and funding instruments, where the dynamics of the underlying asset price include a stochastic component which jumps whenever a counterparty defaults; as a consequence, asset prices and credit events are not independent, and we control their correlation through the jump size. With respect to the paper by \citet{MercurioLi2015}, we generalize their model by including two distinct sources of jump risk, one for each counterparty involved in the transaction. We derive a non-linear PDE associated with the valuation of vulnerable claims in the general case, and we obtain a linearised version under a suitable funding policy which generalizes the assumptions from \citet{BurgardKjaer2011}. Subsequently we apply our framework to study the valuation of a forward contract between two counterparties, both of which are subject to default risk and where, contrary to \citet{BrigoBuescuRutkowski2017}, we assume positive recovery at the time of liquidation. We find the vulnerable forward contract is valued as a continuous portfolio of European call and put options priced under an equivalent probability measure $\widehat{\mathbb{Q}}$ (Proposition \ref{prop:valuation-forward}). Importantly, we show this valuation formula admits an analytical expression involving elementary functions and the Gaussian cumulative distribution function (Theorem \ref{th:analytical-fwd-value}). Finally, numerical analysis demonstrates that neglecting WWR leads to material mispricing of vulnerable claims as well as underestimating the sensitivity to interest rates.
\newline

Our results are relevant from both academic and industrial perspectives. The analytical pricing formula \eqref{eq:analytical-fwd-value} extends the seminal framework from \citet{BlackScholes1973} by incorporating features that have become critical in derivative valuation since the 2007-2008 GFC, including wrong way risk and price-credit correlation. To the best of our knowledge, no such closed-form results have been published in the literature in recent years, where focus has often be on either studying a single deviation from the Black-Scholes framework or instead relying on numerical methods for obtaining prices. Compared to the recent work by \citet{BrigoBuescuRutkowski2017} which study options only, we start off by studying claims with bilateral cash-flows (such as forwards) which exhibit the added complexity of two-way default risk and cash balances with can be both positive and negative. From a practical perspective, our valuation formula is quick to evaluate and provides broker-dealers with swift estimates of the overall cost from entering into such trades. These estimates can be used to quote all-inclusive prices to external clients or estimate parameter sensitivities while avoiding full-blown CVA calculations, which are longer to execute in IT systems and hence vulnerable to material market moves in the meantime. Banks have an active interest in such accurate and speedy valuations to remain price-competitive while ensuring appropriate risk management, particularly in electronic trading \citep{Tunstead2023}, as well as to efficiently allocate limited computing resources across the company.
\newline

The rest of this paper is organized as follows. In Section \ref{sec:valuation-problem-vulnerable-derivatives} we introduce our financial market model and formulate the valuation problem for a vulnerable derivative. In Section \ref{sec:pdes-associated-valuation-problem} we use hedging arguments to derive a non-linear partial differential equation for the pre-default valuation of a vulnerable claim (Proposition \ref{prop:pricing-pde}); we establish conditions under which this PDE becomes linear and obtain a probabilistic representation of the solution (Theorem \ref{th:conditional-expectation}). In Section \ref{sec:analytical-valuation-forward} we specialize our results to the case of a forward contract on a stock, obtaining an analytical formula for valuation (Theorem \ref{th:analytical-fwd-value}). In Section \ref{sec:numerical-analysis} we perform numerical analysis by comparing against risk-free valuations and estimating sensitivities to model parameters (including cross-effects between parameters), before concluding in Section \ref{sec:conclusions}.
\newline

Unless stated otherwise, we adopt the point of view of the broker-dealer and define cash flows so that a positive (resp. negative) cash flow represents an inflow (resp. outflow) for the dealer and an outflow (resp. inflow) for the client. We use the following notation for any numbers $x,y\in\mathbb{R}$ and any set $E\subseteq\mathbb{R}$:
\begin{flalign*}
\qquad x\vee y 			& :=\max\{x,y\}		&&\\
\qquad x\wedge y		& :=\min\{x,y\}		&&\\
\qquad (x)^+			& :=x\vee0			&&\\
\qquad (x)^-			& :=-(x\wedge0)		&&\\
\qquad \indicatorf{E}{x}& :=1 \text{ if $x\in E$, 0 otherwise} &&\\
\qquad \indicator{x\in E}&:=\indicatorf{E}{x} &&\\
\qquad \sgn(x)			& :=\indicatorf{(0,\infty)}{x}-\indicatorf{(-\infty,0)}{x}
\end{flalign*}

\section{The valuation problem for vulnerable derivatives}
\label{sec:valuation-problem-vulnerable-derivatives}

In this section we posit a financial market model with two defaultable counterparties, a derivative dealer and a client, as well as asymmetric interest rates. The two parties want to enter into a contract written on a stock share, whose price jumps when either counterparty defaults. After we describe admissible trading strategies and necessary conditions under which there is no arbitrage, we formulate the dealer's valuation problem for the vulnerable claim.

\subsection{The financial market model}
\label{sec:model-description}
We fix a probability space $(\Omega,\mathscr{F},\mathbb{P})$ over a continuous trading horizon $[0,\bar{T}]$ for some large real $\bar{T}\in\mathbb{R}_{+}$ where $\mathbb{P}$ is the physical probability measure. We consider two market participants: a derivative dealer (counterparty 1) and a corporate client (counterparty 2). Both counterparties can default at any given time $t$ between 0 and $\bar{T}$. To model default we introduce two independent Poisson process $N^{1},N^{2}:[0,\bar{T}]\times\Omega\rightarrow\mathbb{N}$ with intensities $\gamma_1,\gamma_2\in\mathbb{R}_{+}$, and define the \emph{default process} $J^{i}:[0,\bar{T}]\times\Omega\rightarrow\{0,1\}$ for counterparty $i\in\{1,2\}$ as follows:
\begin{align}
	\label{eq:default-def}
	J^{i}_t:=\indicatorf{[1,\infty)}{N^{i}_t}
\end{align}
Because the two Poisson processes are independent, they a.s. never jump at the same time, hence simultaneous defaults do not occur in this model -- the quadratic covariation between the default processes is zero. 
\newline

The default processes generate the default filtration $\mathbb{H}$ and each default process $J^{i}$ induces the $\mathscr{H}_t$-stopping time $\tau_i$ corresponding to the first jump time from the Poisson process $N^{i}$. We define the \emph{first default time} as the stopping time $\tau:=\tau_1\wedge\tau_2$ which generates the \emph{first-to-default jump process} $J:[0,\bar{T}]\times\Omega\rightarrow\{0,1\}$:
	\begin{align}
		\label{eq:first-to-default}
		J_t:=\indicatorf{[\tau,\infty)}{t}
	\end{align}

We introduce a frictionless financial market $\mathfrak{M}$ defined as a set of assets which can be traded by the derivative dealer such that each asset is characterized by a price process and a dividend process as in \citet{Duffie2001}. We equip $\mathfrak{M}$ with eight assets in total: (1) a stock share; (2-3) two defaultable zero-coupon bonds with zero recovery and expiries $T_1,T_2\in(0,\bar{T}]$, issued by counterparties 1 and 2 respectively; (4-6) repurchase agreements (repos) on the stock and the two bonds; (7) a risk-free deposit account to lend cash at a risk-free rate $r_\ell\in\mathbb{R}$; and (8) a funding account to borrow cash on an unsecured basis at a funding rate $r_b\in\mathbb{R}$. 
\newline

The prices of the deposit and funding accounts available to the dealer (whose point of view we are taking) are given by the stochastic processes $B^\ell,B^b:[0,\bar{T}]\times\Omega\rightarrow\mathbb{R}$ with dynamics:
	\begin{align}
		\label{eq:deposit-model}
		& \diff B^\ell_t = r_\ell B^\ell_t\diff t\\
		\label{eq:funding-model}
		& \diff B^b_t = r_b B^b_t\diff t,
	\end{align}
and initial condition $B^\ell_0,B^b_0\in\mathbb{R}_{+}^{*}$ respectively. The defaultable zero-coupon bonds are fully characterized by their price process $P^{i}:[0,\bar{T}]\times\Omega\rightarrow\mathbb{R}$ for each $i\in\{1,2\}$, whose dynamics are given by:
	\begin{align}
		\label{eq:zcbond-model}
		& \diff P^{i}_t 
		= P^{i}_{t-}(r_i\diff t-\diff J^{i}_t),
	\end{align}
with $P^{i}_0=e^{-r_iT_i}>0$ and where $r_i\in\mathbb{R}$ is the rate of return for the zero-coupon bond issued by counterparty $i$. We make the assumption that the rate of return from the bonds exactly compensates the credit risk taken on by the debt holders with respect to a risk-free deposit so that $\gamma_i$ is interpreted as the credit spread from counterparty $i$:
	\begin{align}
		\label{eq:bond-return-rate}
		r_i=r_\ell+\gamma_i
	\end{align}

The stock price is modelled by the process $S:[0,\bar{T}]\times\Omega\rightarrow\mathbb{R}$ such that:
	\begin{align}
		\label{eq:stock-price-process}
		\diff S_t 
		=S_{t-}((\mu-q)\diff t+\sigma\diff W_t
		+\kappa(\diff J^{1}_t+\diff J^{2}_t))
	\end{align}
with initial condition $S_0=s\in\mathbb{R}_{+}^{*}$ and where $W:[0,\bar{T}]\times\Omega\rightarrow\mathbb{R}$ is a Brownian motion; $\mu\in\mathbb{R}$ the stock price drift; $q\in\mathbb{R}$ the stock dividend rate; $\sigma\in\mathbb{R}_{+}^{*}$ the stock price volatility; and $\kappa\in[-1,1]$ the relative price change at time of the first default. We assume the Brownian motion and the Poisson processes $N^{1}$, $N^{2}$ are independent.
\newline

Our stock price model generalizes \citet{MercurioLi2015} where the asset price is only sensitive to the default from one of the parties. By including the default processes into the dynamics of $S$, we achieve two aims: first we correlate the stock price with the default event; secondly we capture gap risk arising from a sudden jump in the stock price at default.
\newline

The correlation between stock price and default event admits a closed-form formula which depends on the jump size $\kappa$, the volatility $\sigma$ and the \emph{first-to-default hazard intensity} $\upgamma:=\gamma_1+\gamma_2$. An approximation to the correlation value for small values of $\kappa$ is, in this model:
	\begin{align}
		\label{eq:rho-linear-kappa}
		\textrm{Corr}(S_t,J_t)=
		\kappa\left(\frac{e^{-\upgamma t}(1-e^{-\upgamma t})}{e^{\sigma^2t}-1}\right)+\mathcal{O}(\kappa^2),
	\end{align}
that is stock-credit correlation is linear in the jump parameter. Table \ref{tab:price-credit-correlation} displays numerical examples when correlation is estimated using the approximation \eqref{eq:rho-linear-kappa} for different values of $\upgamma$ and $\kappa$ under the assumption that $\sigma=15\%$ which is an usual level for implied volatilities from the S\&P 500 US equity index.%
\newline

\begin{table}[h!]
\centering\small
\rowcolors{3}{}{gray!10}
\begin{threeparttable}
	\begin{tabular}{lcccc}
		\toprule
		\multirow{2}{*}{\bf\bm{$\kappa$}} 
		& \multicolumn{4}{c}{\bf\bm{$\upgamma$}}	\\\cmidrule(lr){2-5}
		~    		& \it 0.02 		& \it 0.04	& \it 0.06	& \it 0.1 	\\\midrule
		\it -0.15 	& -12.8\% 		& -16.7\% 	& -19.0\% 	& -21.2\%	\\ 
		\it -0.05 	& -4.3\% 		& -5.6\% 	& -6.3\% 	& -7.1\%	\\
		\it 0.05 	& 4.3\% 		& 5.6\% 	& 6.3\% 	& 7.1\%		\\
		\it 0.15 	& 12.8\% 		& 16.7\% 	& 19.0\%  	& 21.2\%	\\
		\bottomrule
	\end{tabular}
	\caption{Sample correlations between stock price $S_t$ and first-to-default process $J_t$ when $\sigma=0.15$.}
	\label{tab:price-credit-correlation}
\end{threeparttable}
\end{table}

In addition to the price process \eqref{eq:stock-price-process}, the stock share is also characterized by a dividend process $D:[0,\bar{T}]\times\Omega\rightarrow\mathbb{R}$ with initial condition $D_0=0$ which verifies:
\begin{align}
	\label{eq:dividend-model}
	\diff D_t = q S_{t-}\diff t,
\end{align}

Finally, repurchase agreements are defined as pure dividend assets with zero price; the dividend process $R^S:[0,\bar{T}]\times\Omega\rightarrow\mathbb{R}$ for the repo on the stock has the dynamics:
	\begin{align}
		\label{eq:stock-repo-model}
		\diff{R}^S_t=\diff (S_t+D_t) - h_SS_{t-}\diff t,
	\end{align}	
with initial condition $R^S_0=0$ and where $h_S\in\mathbb{R}$ is the repo rate on the stock; similarly the bond repos have dividend process $R^{i}:[0,\bar{T}]\times\Omega\rightarrow\mathbb{R}$ for $i\in\{1,2\}$:
\begin{align}
	\label{eq:zcbond-repo-model}
	\diff{R}^{i}_t=\diff P^{i}_t - h_iP^{i}_{t-}\diff t,
\end{align}	
with $R^{1}_0,R^{2}_0=0$ and where $h_1,h_2\in\mathbb{R}$ are the repo rates.

\begin{remark}
	While we assume the risky asset to be an equity share, the log-normal model \eqref{eq:stock-price-process} is readily applicable to FX products if we interpret $S$ as the exchange rate, $\mu$ as the foreign interest rate and $q$ as the domestic one \citep{GarmanKohlhagen1983}. Log-normal dynamics have also been used in the literature to model commodities \citep{GibsonSchwartz1990} and inflation \citep{JarrowYildirim2003} derivatives.
\end{remark}

The probability space $(\Omega,\mathscr{F},\mathbb{P})$ is equipped with a market filtration $\mathbb{F}$ generated by the price and dividend processes from traded assets in $\mathfrak{M}$. We progressively enlarge this market filtration with the default filtration $\mathbb{H}$ in order to obtain the enlarged filtration $\mathbb{G}$ defined as $\mathbb{G}:=\mathbb{F}\vee\mathbb{H}$. We work under the $(\mathcal{H})$ hypothesis which states that any square-integrable $\mathbb{F}$-martingale remains a square-integrable $\mathbb{G}$-martingale \citep{JeanblancYorChesney2009}.

\subsection{Trading strategies and no-arbitrage in the financial market}
\label{sec:trading-strategies-and-no-arbitrage}
At any time $t\in[0,\bar{T}]$ and for any contingent claim he needs to hedge, the dealer holds a portfolio of traded assets from $\mathfrak{M}$ according to a \emph{trading strategy}, which is defined as a $\mathbb{F}$-predictable stochastic process $\Theta:[0,\bar{T}]\times\Omega\rightarrow\mathbb{R}^8$ specifying the number of units held of each asset such that:
	\begin{align}
		\label{eq:trading-strategy}
		\Theta_t:=(
		\theta^S_t,
		\theta^{1}_t,
		\theta^{2}_t,
		\upvartheta^S_t,
		\upvartheta^{1}_t,
		\upvartheta^{2}_t,
		\theta^\ell_t,
		\theta^b_t
		)
	\end{align}
where the elements of $\Theta$ are themselves $\mathbb{F}$-predictable stochastic processes from $[0,\bar{T}]\times\Omega$ into $\mathbb{R}$ which give respectively the units of stock $\theta^S$; defaultable bonds from counterparties one $\theta^{1}$ and two $\theta^{2}$; repos on the stock and the two bonds $\upvartheta^S$, $\upvartheta^{1}$ and $\upvartheta^{2}$; the deposit account $\theta^\ell$; and the funding account $\theta^b$. 
\newline

The value of the portfolio constructed with the strategy $\Theta$ is a stochastic process $V^\Theta:[0,\bar{T}]\times\Omega\rightarrow\mathbb{R}$ defined as the dot product of the price and trading strategy vectors:
	\begin{align}
		\label{eq:portfolio-price}
		V^\Theta_t:=
		\theta^S_tS_t
		+\sum_{i=1}^2\theta^{i}_tP^{i}_t
		+\theta^\ell_tB^\ell_t
		+\theta^b_tB^b_t
	\end{align}
where we have used the fact that repos are zero price assets. The \emph{gain process} associated to this portfolio is the stochastic process $G^\Theta:[0,\bar{T}]\times\Omega\rightarrow\mathbb{R}$ defined as follows \citep[see][]{Duffie2001}:
	\begin{align}
		\notag
		G^\Theta_t&:=
		\int_0^t\theta^S_u(\diff S_u+\diff D_u)
		+\int_0^t\upvartheta^S_u\diff R^S_u
		+\int_0^t\sum_{i=1}^2\theta^{i}_u\diff P^{i}_u
		+\int_0^t\sum_{i=1}^2\upvartheta^{i}_u\diff R^{i}_u
		\\[3pt]\label{eq:portfolio-gain-process}
		& \qquad 
		+\int_0^t\theta^\ell_u\diff B^\ell_u
		+\int_0^t\theta^b_u\diff B^b_u
	\end{align}

We impose a series of conditions on the trading strategies. First while the dealer can take long (i.e. positive number of units) or short (negative) positions in risky assets, the deposit account can only be used for positive cash balances while the funding account must be used for negative ones. Therefore for all $t\in[0,\bar{T}]$ we impose the \emph{funding condition}:
\begin{align}
	\label{eq:bank-account-constraint}
	\theta^\ell_t\geq0, \qquad
	\theta^b_t\leq0
\end{align}
Secondly, it is easy to prove that for any trading strategy such that $\theta^{\ell,1}\theta^{b,1}\neq0$ there exists a strategy $\theta^{\ell,0}\theta^{b,0}\equiv0$ such that the portfolio value is equal in both cases but the latter has higher cumulative gains. Therefore we impose the \emph{netting condition} for all $t\in[0,\bar{T}]$ which states we do not simultaneously lend and borrow cash \citep{Mercurio2014}:
	\begin{align}
		\label{eq:netting-condition}
		\theta^\ell_t\theta^b_t=0
	\end{align}
The value of trading strategies is restricted to be bounded from below by a real number $A\geq0$ \citep{JeanblancYorChesney2009}:
	\begin{align}
		V^\Theta_t\geq -A
	\end{align}
for all $t\in[0,\bar{T}]$. We say a trading strategy $\Theta$ is \emph{admissible} if it is bounded from below and satisfies both the funding and netting conditions.
\newline

We end this section by providing necessary conditions for absence of arbitrage in market $\mathfrak{M}$. We recall from \citet{Duffie2001} that a strategy is \emph{self-financing} if the next equality holds for all $t$:
	\begin{align}
	\label{eq:def-self-financing}
	G^\Theta_t=V^\Theta_t-V^\Theta_0
	\end{align}
We define an \emph{arbitrage} strategy along the lines of \citet{Bjork2020}. It is a self-financing trading strategy $\Theta$ such that for some $t\in[0,\bar{T}]$:
	\begin{align}
		\label{eq:arb-def}
		V^\Theta_0=0,\qquad \mathbb{P}(G^\Theta_t\geq0)=1,\qquad \mathbb{P}(G^\Theta_t>0)>0
	\end{align}
We state the following proposition.
\begin{proposition}
	\label{prop:no-arbitrage}
	If the financial market $\mathfrak{M}$ is \emph{arbitrage-free} then the following must hold:
		\begin{align}
			\label{eq:funding-arbitrage}
			\forall\; h\in\{h_S,h_1,h_2\},\ & r_\ell\leq h\leq r_b\\[4pt]
			\label{eq:credit-arbitrage}
			\forall\; i\in\{1,2\},\ &h_i<r_i
		\end{align}	
\end{proposition}
\begin{proof}
	See Appendix \ref{app:proof-no-arbitrage}.
\end{proof}

Immediate consequences are that under no-arbitrage: $r_\ell\leq r_b$ and $0<\gamma_i$ for any $i$ where the second inequality is a consequence from Proposition \eqref{prop:no-arbitrage} and assumption \eqref{eq:bond-return-rate}. In practice we expect the borrowing rate $r_b$ for the dealer to be close to its own bond rate $r_i$.
\newline

Under conditions \eqref{eq:funding-arbitrage}-\eqref{eq:credit-arbitrage} the prices of risky assets (stock, bonds, repos) have neither a.s. increasing trajectories nor a.s. decreasing ones \citep[see Proposition 9.9 in][on no-arbitrage conditions in L\'evy market models]{ContTankov2004}. 
\newline

For the remaining of the paper we assume that $\mathfrak{M}$ is arbitrage-free as well as \emph{complete}, meaning that any $\mathscr{F}_T$-mesurable contingent cash flow $X:\Omega\rightarrow\mathbb{R}$ for $T\in(0,\bar{T}]$ can be replicated using the traded assets.
\newline

Having described the financial market as well as the trading strategies available to the broker-dealer, we can now formulate the valuation problem for a derivative claim.

\subsection{Statement of the valuation problem}
\label{sec:problem-statement}
We assume counterparty 2 (the client) wants to trade with counterparty 1 (the dealer) a \emph{vulnerable} European derivative claim written on the stock price $S$ and with expiry $T\in(0,\bar{T})$ such that $T\leq T_1,T_2$ -- where vulnerable means that the claim is subject to the risk that any party might default.
\newline

The claim design is as follows:
	\begin{itemize}
		\item If $\tau>T$, both parties have survived until expiry and the claim pays the cash amount ${{f}}(S_T)$ at time $T$ where ${{f}}:\mathbb{R}_{+}^{*}\rightarrow\mathbb{R}$ is the \emph{terminal payoff} function. We assume the claim has bilateral payoffs meaning the sign of the payoff function ${{f}}$ can be both negative and positive. 
		\item If $\tau\leq T$, one party has defaulted before expiry and the claim is terminated after paying out an amount which depends on its \emph{mark-to-market} (MTM) value, which is an $\mathbb{F}$-adapted process $M:[0,T]\times\Omega\rightarrow\mathbb{R}$ that admits the representation $M_t=m(t,S_t)$ where $m:[0,T]\times\mathbb{R}_{+}^{*}\rightarrow\mathbb{R}$ is a continuous function. 
	\end{itemize}

The MTM value is usually independent from the counterparties transacting the claim \citep[see][]{BurgardKjaer2011} and serves as a benchmark against which to calculate the payout at default, but the paid amount normally depends on whether either the creditor or the debtor party has defaulted first: generally speaking, a creditor cannot expect to recover the full MTM value from its debtor. 
\newline

The \emph{recovery} or \emph{close-out amount} represents the actual amount which is either paid or received by the dealer and which depends on the MTM value at default. The ISDA 2002 Master Agreement which governs bilateral derivative transactions explains recovery is usually done either at full MTM value if the claim has positive value to the defaulting party, at a fraction of the MTM otherwise. 
\newline

To model recovery, we define the process $\wthalf{M}:[0,T]\times\Omega \rightarrow \mathbb{R}$ using a function $\widetilde{m}:[0,T]\times\mathbb{R}_+^*\rightarrow\mathbb{R}$:
	\begin{align}
		\label{eq:M-tilde}
		\wthalf{M}_t
		=\widetilde{m}(t,S_{t-})
		:=m(t,(1+\kappa)S_{t-})
	\end{align}
which is constructed to be left-continuous and therefore $\mathbb{F}$-predictable. We introduce recovery processes $Z^{1}, Z^{2}:[0,T]\times\Omega\rightarrow\mathbb{R}$ and recovery rates $\varkappa_1,\varkappa_2\in(0,1]$ where $Z^{i}$ (resp. $\varkappa_i$) represents the amount (resp. fraction) recovered (or paid out) by party $i$ in case it defaults first. Formally, we require the recovery process to be $\mathbb{F}$-predictable as in \citet{BieleckiJeanblancRutkowski2004} hence we posit the following specification which is similar to that from \citet{BurgardKjaer2011}:
	\begin{align}
		\label{eq:recovery-value-1}
		Z_t^{1}&:=(\wthalf{M}_t)^+-\varkappa_1(\wthalf{M}_t)^-
		\\[3pt]
		\label{eq:recovery-value-2}
		Z_t^{2}&:=\varkappa_2(\wthalf{M}_t)^+-(\wthalf{M}_t)^-
	\end{align}
The recovery values can be modelled as $Z^{i}_t=z_i(t,(1+\kappa)S_{t-})$ where $z_i:[0,T]\times\mathbb{R}_{+}^{*}\rightarrow\mathbb{R}$ is a continuous function, so that the claim is terminated by paying off the close-out amount $z_i(\tau_i,(1+\kappa)S_{\tau_i-})=z_i(\tau_i,S_{\tau_i})$ if $\tau=\tau_i$. By setting $\varkappa_i>0$ we assume recovery is strictly positive.
\newline

In the following, we will use the notation $\widehat{\mathcal{C}}_1(f,z_1,z_2)$ to refer to this claim. This notation stresses that, under our assumptions, the claim is fully specified by its terminal payoff function ${{f}}$ and the recovery pair $(z_1,z_2)$; moreover, cash flows are defined with respect to counterparty 1. If there is no ambiguity, we will shorten to $\widehat{\mathcal{C}}_1$.
\newline

Based on our description of $\widehat{\mathcal{C}}_1$, we choose to characterize the claim cash flows by introducing value, dividend and gain processes $\widehat{V}, \widehat{D}, \widehat{G}:[0,T]\times\Omega\rightarrow\mathbb{R}$, respectively, defined as follows:
	\begin{subequations}
	\label{eq:def-contract}
	\begin{alignat}{2}
		\label{eq:def-contract-price}
		&\widehat{V}_t:=
		\indicator{t<\tau}
		V_t,&&\qquad \widehat{V}_0=V_0\\[4pt]
		\label{eq:def-contract-dividend}
		&\widehat{D}_t:=\sum_{i=1}^2\int_{0}^t\left(
		\indicator{u\leq\tau}
		Z^i_u\right)\diff
		\indicator{u\geq\tau_i},
		&&\qquad \widehat{D}_0=0\\[4pt]
		\label{eq:def-contract-gains}
		&\widehat{G}_t:=(\widehat{V}_t-\widehat{V}_0)+\widehat{D}_t,&&\qquad \widehat{G}_0=0
	\end{alignat}
	\end{subequations}
where $V:[0,T]\times\Omega\rightarrow\mathbb{R}$ in \eqref{eq:def-contract-price} is the \emph{pre-default value} of the claim. All price process in $\mathfrak{M}$ have the Markov property under $\mathbb{P}$ thus we assume $V$ admits the following representation for $t\in[0,T]$:
	\begin{align}
		V_t:=v(t,S_t)
	\end{align}
for a twice continuously differentiable function $v:[0,T]\times\mathbb{R}_{+}^{*}\rightarrow\mathbb{R}$, which we call the \emph{pre-default valuation function}. The dividend process is defined as the sum of integrals of the predictable process $\mathbbm{1}_{[0,\tau]}Z^{i}$ with respect to the \cadlag\ jump process $J^{i}$ for $i=1,2$.
\newline

Similar modelling approaches have already been suggested in the literature. The approach we adopt here is close to that from \citet{BieleckiRutkowski2004} although the latter justify the representation \eqref{eq:def-contract} using martingale methods. 
\newline

Alternatively, our characterization is equivalent to defining a single (price) process $\widehat{V}^\prime$ through a function $\widehat{v}:[0,T]\times\mathbb{R}_{+}^{*}\times\{0,1\}^2\rightarrow\mathbb{R}$ such that $\widehat{G}=\widehat{V}^\prime$ and:
	\begin{subequations}
	\begin{align}
		\label{eq:BK-notation}
		& \widehat{V}^\prime_t=\indicator{t\leq\tau}\widehat{v}(t,S_t,J^{1}_t,J^{2}_t)\\[2pt]
		& \widehat{v}(t,s,0,0):=v(t,s)\\[2pt]
		& \widehat{v}(t,s,1,0):=z_1(t,(1+\kappa)s)\\[2pt]
		& \widehat{v}(t,s,0,1):=z_2(t,(1+\kappa)s)
	\end{align}
	\end{subequations}
which is in the spirit of \citet{BurgardKjaer2011}.

\begin{remark}
	It is possible to express $\widehat{V}$ as a function of all asset prices in $\mathfrak{M}$. However, as noted in Section 7.9.2 of \citet{JeanblancYorChesney2009}, we can restrict ourselves to a function of $t$ and $S$ only because the bank accounts and the non-defaultable components of the zero-coupon bonds are deterministic, and any such dependency is captured by the time variable $t$.
\end{remark}

\begin{remark}
	\citet{BieleckiRutkowski2004} define the value from the recovery dividend $\widehat{D}$ assuming that it gets reinvested in the deposit account $B^\ell$ after default. In practice, the claim is terminated at $\tau$ and we can neglect the trading strategy afterwards.
\end{remark}

Using the definition of stochastic integrals with respect to jump processes, the gains $\widehat{G}$ admit the representation:
	\begin{align}
		\widehat{G}_t+V_0
		=\indicator{t<\tau}V_t+\sum_{i=1}^2\left(\indicator{(t\wedge\tau_i)\leq\tau}Z^i_{t\wedge\tau_i}\right)\indicator{t\geq\tau_i}
		=\indicator{t<\tau}V_t+\sum_{i=1}^2\left(\indicator{\tau_i\leq\tau}Z^i_{\tau_i}\right)\indicator{t\geq\tau_i}
	\end{align}
Therefore the \emph{loss-given-default} (LGD) $\Delta\widehat{G}:[0,T]\times\Omega\rightarrow\mathbb{R}$, defined for any $t$ as $\Delta \widehat{G}_t:=\widehat{G}_t-\widehat{G}_{t-}$, satisfies at $\tau$:
	\begin{align}
		\Delta\widehat{G}_\tau
		=\sum_{i=1}^2\indicator{\tau=\tau_i}Z^{i}_{\tau}-V_{\tau-}
	\end{align}
The gain process as defined in \eqref{eq:def-contract-gains} exhibits the characteristics we want in a valuation model with credit and wrong way risks. Before default, the claim has a certain value $V$ which fluctuates through time based on the changes in the stock price $S$. At default, the value $V_{\tau-}$ of the claim prior to any jump is extinguished, while a recovery dividend $Z^i$ is paid out after the stock price $S$ has jumped -- where $i$ is the defaulting party.
\newline

The objective of counterparty 1 is to hedge the claim by replicating the gains $\widehat{G}$ generated by $\widehat{\mathcal{C}}_1$ over $[0,T\wedge\tau]$. Assuming the recovery functions $z_1,z_2$ are exogenous, the problem reduces to finding the pre-default value $V$, and we can now formally state the dealer valuation problem.
\begin{problem}
	\label{problem}
	Does the pre-default valuation function $v$ for the vulnerable bilateral claim $\widehat{\mathcal{C}}_1({{f}},z_1,z_2)$ -- under credit, funding and wrong way risk as well as positive recovery -- admit an analytical expression?
\end{problem}

\section{PDEs induced by the valuation problem for vulnerable contracts}
\label{sec:pdes-associated-valuation-problem}

Now that we have formulated the valuation problem \ref{problem}, we use hedging arguments to derive PDEs which characterize the valuations of such contracts.

\subsection{Non-linear PDE for the general case}
\label{sec:pde-general-case}
To derive an expression for the pre-default function $v$, we construct a self-financing trading strategy $\Theta$ such that the portfolio induced by $\Theta$ has same value $\widehat{V}$ as the claim up to default while hedging the recovery dividend $\widehat{D}$.
\newline

Formally, let us consider a hedged portfolio where we hold a single unit of claim $\widehat{\mathcal{C}}_1$ with value $\widehat{V}$ together with assets from $\mathfrak{M}$ in quantities given by a trading strategy $\Theta$. For $\Theta$ to hedge the claim up to default included, we need to ensure the strategy's valuation matches that from the claim:
	\begin{align}
		\label{eq:hedging-condition-price}
		V^\Theta_t+\widehat{V}_t=0
	\end{align}
for any $t\in[0,T]$; we refer to \eqref{eq:hedging-condition-price} as the \emph{hedging condition} for claim $\widehat{\mathcal{C}}_1$. We also require $\Theta$ to be self-financing i.e. $G^\Theta_t=V^\Theta_t-V^\Theta_0$ up to default%
	\footnote{Note at $\tau$ the portfolio is trivially self-financing because the only gain arises from the price jump.}. 
From the definition of $\widehat{G}$, observe that $\widehat{G}_t=\widehat{V}_t-\widehat{V}_0$ over $\{t<\tau\}$; moreover if \eqref{eq:hedging-condition-price} holds then $V^\Theta_t-V^\Theta_0=\widehat{V}_t-\widehat{V}_0$ for any $t$. Hence we formulate the following \emph{self-financing condition}:
	\begin{align}
		\label{eq:self-financing-condition-gains}
		G^{\Theta}_t+\widehat{G}_t=0
	\end{align}
for any $t\in[0,T]$. We say a trading strategy $\Theta$ \emph{hedges} claim $\widehat{\mathcal{C}}_1$ if it is admissible and satisfies both conditions 
	\eqref{eq:hedging-condition-price} and
	\eqref{eq:self-financing-condition-gains}.
\newline

We define the \emph{exposures} to the stock $\xi^S$ and the bonds $\xi^{i}$ as the total units in the physical asset and its repo for any $t\in[0,T]$:
	\begin{align}
		\label{eq:def-exposures}
		\xi^S_t:=\theta^S_t+\upvartheta^S_t,\qqquad
		\xi^{i}_t:=\theta^{i}_t+\upvartheta^{i}_t
	\end{align}
The following proposition characterizes the \emph{hedging strategy} for the vulnerable claim.

\begin{proposition}
	\label{prop:hedging-strategy}
	The hedging strategy $\Theta$ for the vulnerable claim $\widehat{\mathcal{C}}_1(f,z_1,z_2)$ with pre-default valuation $V$ and dividend process $\widehat{D}$ satisfies for $i\in\{1,2\}$ and $t\in[0,T]$:
	\begin{align}
		\label{eq:stock-hedging}
		&\xi^S_t
		=-\indicator{t\leq\tau}v_s(t,S_{t-})
		\\[11pt]\label{eq:bond-hedging}
		&\xi^{i}_t
		=\frac{\indicator{t\leq\tau}}{P^{i}_{t-}}
		\left(z_i(t,(1+\kappa)S_{t-})-v(t,S_{t-})-v_s(t,S_{t-})\kappa S_{t-}\right)
		\\\label{eq:deposit-hedging}
		&\theta^\ell_t
		=\frac{1}{B^\ell_t}\left(
		-\indicator{t\leq\tau}V_{t-}-\theta^S_tS_{t-}
		-\sum_{i=1}^2\theta^{i}_tP^{i}_{t-}\right)^+
		\\\label{eq:funding-hedging}
		&\theta^b_t
		=-\frac{1}{B^b_t}\left(
		-\indicator{t\leq\tau}V_{t-}-\theta^S_tS_{t-}
		-\sum_{i=1}^2\theta^{i}_tP^{i}_{t-}\right)^-
	\end{align}
\end{proposition}
\begin{proof}
	See Appendix \ref{app:proof-lemma-hedging-strategy}.
\end{proof}

We introduce the \emph{funding policy} $\upalpha=(\alpha_S,\alpha_1,\alpha_2)\in[0,1]^3$ which represents the exposure split between outright purchases and repos for the stock and bonds, where the vector $\upalpha$ represents the fraction of the exposure funded using the bank accounts, assumed to be static across $[0,T]$. Hence in the stock case, we have:
\begin{align}
	\label{eq:funding-policy-stock}
	\theta^S_t=\alpha_S\xi^S_t,\qquad
	\upvartheta^S_t=(1-\alpha_S)\xi^S_t,
\end{align}
An equivalent definition holds for the funding strategy $\alpha_1, \alpha_2$ of the bonds. We define the standard \emph{pricing operator} $\mathscr{L}_S$ for process $S$ as follows \citep[][]{Piterbarg2010}:
\begin{align}
	\label{eq:pricing-operator}
	\mathscr{L}_Sv(t,s)
	:=\left(\frac{\partial}{\partial t}+\frac{\sigma^2s^2}{2}\frac{\partial^2}{\partial s^2}\right)v(t,s)
\end{align}

\begin{proposition}
	\label{prop:pricing-pde}
	The pre-default valuation function $v:(t,s)\mapsto v(t,s)$ can be characterized as the solution to the following Cauchy problem:
		\begin{multline}
			\label{eq:pricing-pde}
			-r_Vv
			+\mathscr{L}_Sv
			+(r_S-q-\kappa\lambda_V)v_ss
			+\sum_{i=1}^2\lambda_iz_i\\
			-\varphi\left((\alpha_P-1)v+(\alpha_S+\alpha_P\kappa)v_ss
			-\sum_{i=1}^2\alpha_{i}z_i\right)^-
			=0
		\end{multline}
	with terminal condition $v(T,s)={{f}}(s)$ for $s\in\mathbb{R}_{+}^{*}$ and $t\in[0,T]$, where we have defined the following variables for $i\in\{1,2\}$:
		\begin{align}
		\label{eq:def-rs}
		r_S
		&:=\alpha_Sr_\ell+(1-\alpha_S)h_S\\[3pt]
		\label{eq:def-lambdai}
		\lambda_i
		&:=\gamma_i-(1-\alpha_i)(h_i-r_\ell)\\[4pt]
		\label{eq:def-lambdav}
		\lambda_V
		&:=\lambda_1+\lambda_2\\[3pt]
		\label{eq:def-rv}
		r_V
		&:=r_\ell+\lambda_V\\[4pt]
		\label{eq:def-varphi}
		\varphi
		&:=r_b-r_\ell\\[3pt]
		\label{eq:def-alphap}
		\alpha_P
		&:=\alpha_1+\alpha_2
		\end{align}
\end{proposition}

\begin{proof}
	See Appendix \ref{app:proof-prop-pricing-pde}.
\end{proof}

Even before assuming any specific functional form for the recovery functions $z_1$ and $z_2$, the PDE \eqref{eq:pricing-pde} is non-linear due to the presence of the minimum operator around the unknown function $v$ and its spatial partial derivative $v_s$. To the best of our knowledge, problem \eqref{eq:pricing-pde} does not admit a general analytical solution due to this non-linearity \citep[see also][for a discussion]{Piterbarg2015}.

\subsection{PDE for the linearised case and stochastic representation via Feynman-Kac}
\label{sec:pde-linearised-case}

We seek to determine conditions under which the PDE \eqref{eq:pricing-pde} can be linearised with respect to the pre-default function $v$. Let $\alpha\in[0,1]$, we define the following funding policy $\upalpha^\star$:
	\begin{subequations}
	\label{eq:linearising-funding}
	\begin{align}
		\label{eq:linear-bond1-funding}
		\alpha_1^\star&:=\alpha\\[2pt]
		\label{eq:linear-bond2-funding}
		\alpha_2^\star&:=1-\alpha\\[2pt]
		\label{eq:linear-stock-funding}
		\alpha_S^\star&:=-\kappa(\alpha_1^\star+\alpha_2^\star)
	\end{align} 
	\end{subequations}
Under this funding policy $\upalpha^\star$ the following equalities hold:
	\begin{align}
		\label{eq:linearising-funding-consequences}
		r_S=h_S+\kappa(h_S-r_\ell),\qquad
		\lambda_1=\gamma_1-(1-\alpha)(h_1-r_\ell),\qquad
		\lambda_2=\gamma_2-\alpha(h_2-r_\ell)
	\end{align}
The following corollary ensues.

\begin{corollary}
	\label{coro:pricing-pde-linearised}
	Under the assumptions of Proposition \ref{prop:pricing-pde} and the funding policy $\upalpha^\star$ defined in \eqref{eq:linearising-funding}, the pre-default Cauchy problem \eqref{eq:pricing-pde} becomes linear:
	\begin{align}
		\label{eq:pricing-pde-linearised}
		&-r_Vv
		+\mathscr{L}_Sv
		+(h_S-q+\kappa(h_S-r_V))v_ss
		+\varrho_1(\widetilde{m})^+-\varrho_2(\widetilde{m})^-=0
	\end{align}
	with terminal condition $v(T,s)={{f}}(s)$ for $(s,t)\in\mathbb{R}_{+}^{*}\times[0,T]$, where we have defined the following variables:
		\begin{align}
			\varrho_1&:=\lambda_1+\lambda_2\varkappa_2-\varphi(\alpha+(1-\alpha)\varkappa_2)\\[1pt]
			\varrho_2&:=\lambda_1\varkappa_1+\lambda_2
		\end{align}
\end{corollary}

\begin{proof}
	See Appendix \ref{app:proof-coro-pricing-pde-linearised}.
\end{proof}

We have assumed $\alpha_S\geq0$ hence the funding policy $\upalpha^\star$ is only viable when $\kappa\in[-1,0]$ which corresponds to the case where the stock price $S$ falls whenever counterparties 1 and 2 default: this is the most likely (and interesting) case in practice. Therefore we assume $\kappa\leq0$ from now on.
\newline

We call $\upalpha^\star$ the \emph{linearising funding policy}. In particular, the strategy for the bonds $(\alpha_1^\star,\alpha_2^\star)$ generalizes \citet{BurgardKjaer2011} in which the authors consider the specific case $\alpha=1$. Equation \eqref{eq:linear-stock-funding} is equivalent to funding all exposure to the stock price $S$ through repurchase agreements except the default profit-and-loss $\Delta S$. On the other hand, equations \eqref{eq:linear-bond1-funding}-\eqref{eq:linear-bond2-funding} imply half of the total exposure to the defaultable bonds is financed using repos. These assumptions are consistent with standard market practice whereby participants, in particular large derivative dealers, tend to fund hedging portfolios using repos. 
\newline

Assuming the MTM function $m$ is independent from the pre-default function $v$, the PDE \eqref{eq:pricing-pde-linearised} is now linear because the term within the minimum operator does not depend any longer on the unknown function.
\newline

The solution to the linearised Cauchy problem \eqref{eq:pricing-pde-linearised} admits a stochastic representation which is derived using the Feynman-Kac formula. We give such representation in the following theorem, where we use the notation $\mathbb{E}_{t,s}(\cdot):=\mathbb{E}(\cdot|S_t=s)$ for conditional expectations.

\begin{theorem}
	\label{th:conditional-expectation}
	Assume the functions $f$, $m$ and $v$ satisfy a polynomial growth condition for some constants $a>0$ and $b\geq2$:
		\begin{alignat}{2}
			\label{eq:polynomial-growth-f}
			& |{{f}}(x)|\leq a(1+x^{b}),\quad 
			&&x\in\mathbb{R}_{+}^{*}\\[3pt]
			\label{eq:polynomial-growth-m}
			& |m(t,x)|\leq a(1+x^{b}),\quad 
			&&x\in\mathbb{R}_{+}^{*},\enspace t\in[0,T]\\[3pt]
			\label{eq:polynomial-growth-v}
			& |v(t,x)|\leq a(1+x^{b}),\quad 
			&&x\in\mathbb{R}_{+}^{*},\enspace t\in[0,T]
		\end{alignat}
	Then the solution $v$ to the linearised Cauchy problem \eqref{eq:pricing-pde-linearised} is unique and admits the following stochastic representation for any $(t,s)\in[0,T]\times\mathbb{R}_{+}^{*}$:
		\begin{align}
			\label{eq:conditional-expectation}
			v(t,s)&=\mathbb{E}_{t,s}^{\widehat{\mathbb{Q}}}\left(
			e^{-r_V(T-t)}f(S_T)\right)
			+\mathbb{E}_{t,s}^{\widehat{\mathbb{Q}}}\left(
			\int_t^Te^{-r_V(u-t)}\left(
			\varrho_1(\wthalf{M}_u)^+
			-\varrho_2(\wthalf{M}_u)^-
			\right)\diff u\right)
		\end{align}
	under a probability measure $\widehat{\mathbb{Q}}$ 	where $S$ solves the SDE for $t\in[0,T]$:
		\begin{align}
			\label{eq:dynamics-s-pricing-measure}
			\diff S_t=
			S_t((h_S-q+\kappa(h_S-r_V))\diff t+\sigma\diff \widehat{W}_t)
		\end{align}
	 with initial condition $S_0$ and where $\widehat{W}$ is a Brownian motion under measure $\widehat{\mathbb{Q}}$.
\end{theorem}
\begin{proof}
	The conditions for applying the Feynman-Kac formula as stated in Theorem 7.6 from \citet{KaratzasShreve1991} are fulfilled, including our assumptions that $v$ is twice continuous differentiable, and the payoff $f$ and MTM $m$ functions are continuous. In particular, it is evident that if $m$ satisfies the polynomial growth condition \eqref{eq:polynomial-growth-m} then $\varrho_1(\widetilde{m})^+-\varrho_2(\widetilde{m})^-$ does too.
\end{proof}

Theorem \ref{th:conditional-expectation} gives a probabilistic solution to the linearised Cauchy problem for $v$ under a measure $\widehat{\mathbb{Q}}$ where the stock price process $S$ does not jump but where its dynamics are compensated by the term $\kappa(h_S-r_V)S_t\diff t$, which includes the price's expected jump over $\diff t$ under $\mathbb{P}$. We highlight that, contrary to the approach from \citet{BurgardKjaer2011}, we do not make any assumption or \emph{ansatz} on the decomposition of the pre-default valuation $V$ into additive components when deriving formula \eqref{eq:conditional-expectation}.

\begin{proposition}
	\label{prop:equivalent-measures}
	The stock price process $S$ under the physical measure $\mathbb{P}$ is equivalent to the price process under the pricing measure $\widehat{\mathbb{Q}}$ if and only if $\kappa>-1$ i.e. the stock price does not collapse to zero at default.
\end{proposition}
\begin{proof}
	Under measure $\widehat{\mathbb{Q}}$ the price process \eqref{eq:dynamics-s-pricing-measure} is a pure geometric Brownian motion (GBM) with starting value $S_0=s>0$ thus its null set is $\{\omega:S(t,\omega)\leq0\}$. On the other hand, under the measure $\mathbb{P}$ the price process \eqref{eq:stock-price-process} behaves as a GBM which jumps by a percentage $\kappa\in[-1,0]$ whenever parties 1 or 2 default. Hence the process $S$ will remain strictly positive under $\mathbb{P}$ if and only if $\kappa>-1$.
\end{proof}

To ensure $\widehat{\mathbb{Q}}$ is equivalent to $\mathbb{P}$, for the rest of the paper we fix the range of allowable values for the jump size:
	\begin{align}
		\label{eq:kappa-domain}
		\kappa\in(-1,0]
	\end{align}
This assumption ensures the two probabilities are equivalent while maintaining the viability of the linearising funding policy $\alpha^\star_S$ for the stock. However it excludes certain cases where the stock price can fall to zero, e.g. if the shares are from any of the two counterparties 1 or 2.

\begin{remark}
	We refer to \citet{Friedman1964,Friedman1975} for sufficient conditions under which the polynomial growth condition \eqref{eq:polynomial-growth-v} is guaranteed to hold. For claims with \emph{linear} payoff functions, we  expect the valuation to be bounded from above by a polynomial of the stock price, hence Theorem \ref{th:conditional-expectation} should hold for most common payoffs. Once the payoff and MTM functions have been specified, one can calculate the conditional expectation \eqref{eq:conditional-expectation} then verify it satisfies condition \eqref{eq:polynomial-growth-v}.
\end{remark}

\section{Analytical valuation formula for a forward contract}
\label{sec:analytical-valuation-forward}

We now apply the framework developed in Sections \ref{sec:valuation-problem-vulnerable-derivatives} and \ref{sec:pdes-associated-valuation-problem} to the valuation of a vulnerable forward contract, for which the payoff function ${{f}}$ is defined as follows for any $s\in\mathbb{R}_{+}^{*}$:
	\begin{align}
		\label{eq:forward-payoff}
		{{f}}(s):=s-K
	\end{align}
for some strike value $K\in\mathbb{R}_{+}^{*}$.

\subsection{Risk-free close-out}
\label{sec:mtm-assumption}
Until now we have not made any assumption on the mark-to-market process $M_t$ introduced in Section \ref{sec:problem-statement}. We recall that the recovery amount in case of default is determined based on this MTM value.
\newline

We introduce an equivalent \emph{risk-free claim} $\mathcal{C}_1^*(f)$ written on the stock price $S$ and with same payoff function $f$ as the vulnerable claim $\widehat{\mathcal{C}}_1(f,z_1,z_2)$ but with no default risk from the counterparties. Because there is no credit risk from either party, we assume the dealer can hedge this contract simply by trading repos on the stock and that any negative cash balance (resp. positive) can be borrowed (resp. loaned) at a unique \emph{risk-free rate} $r\geq r_\ell$ (for example, a collateral rate).
\newline

Let $V^*_t:=v^*(t,S_t)$ be a stochastic process representing the valuation for $\mathcal{C}_1^*(f)$ where $v^*$ is a continuous function. We say $V^*$ is the \emph{risk-free value} of the vulnerable contract $\widehat{\mathcal{C}}_1$ because we are valuing a claim with same terminal payoff $f$ but in the absence of default or funding asymmetry.
\newline

Using hedging arguments, the risk-free valuation function $v^*$ satisfies the Cauchy problem:
	\begin{align}
	\label{eq:risk-free-pde}
	&-rv^*+\mathscr{L}_Sv^*+(h_S-q)v^*_ss=0
\end{align}
with terminal condition $v^*(T,s)=f(s)$. Notice this is the same PDE as \eqref{eq:pricing-pde-linearised} but with no credit, asymmetric funding or WWR components.
\newline

Assuming suitable regularity conditions on $v^*$, the Feynman-Kac theorem gives a stochastic representation for the solution to problem \eqref{eq:risk-free-pde}:
\begin{align}
	\label{eq:risk-free-expectation}
	&v^*(t,s)=\mathbb{E}^{\mathbb{Q}^*}_{t,s}\left(e^{-r(T-t)}{f}(S_T)\right)
\end{align}
under a measure $\mathbb{Q}^*$ where the dynamics of $S$ are:
\begin{align}
	\label{eq:risk-free-dynamics}
	\diff S_t=S_t((h_S-q)\diff t+\sigma\diff W^*_t)
\end{align}
with initial condition $S_0$, where $W^*$ is a Brownian Motion. In the case where the payoff function $f$ is that of a forward, the risk-free value \eqref{eq:risk-free-expectation} reduces to:
\begin{align}
	\label{eq:risk-free-forward}
	v^*(t,s)=e^{-r(T-t)}\left(F_T(t,s)-K\right)
\end{align}
where we define $F_T:[0,T]\times\mathbb{R}_+^*\rightarrow\mathbb{R}$ as the \emph{risk-free forward price} of the stock for delivery at time $T$ conditional on the price at time $t$ being $s$:
\begin{align}
	\label{eq:rf-fwd-price-stock}
	F_T(t,s)
	:=\mathbb{E}^{\mathbb{Q}^*}_{t,s}(S_T)
	=se^{(h_S-q)(T-t)}
\end{align}
We say a contract, either risk-free or vulnerable, is \emph{at-the-money} (ATM) at time $t$ when its strike is equal to the risk-free forward price:
\begin{align}
	\label{eq:def-atm}
	K=F_T(t,S_t)
\end{align}

For the remainder of this paper, we assume \emph{risk-free close-out} (or recovery) that is the mark-to-market process $M$ is equal to the risk-free value $V^*$:
	\begin{align}
		\label{eq:risk-free-mtm-assumption}
		m\equiv v^*
	\end{align}
This approach is consistent with ISDA requirements that MTM values should be established without reference to any of the parties involved in the deal.
\newline

Under risk-free close-out, we are in a position to derive a representation for the price of a vulnerable forward contract.

\begin{remark}
	\label{remark:measures-coincide}
	If $\alpha_S=\alpha_S^\star$ and $\kappa=0$ then measures $\mathbb{Q}^*$ and $\widehat{\mathbb{Q}}$ coincide as in \citet{BurgardKjaer2011}. In particular, when $\kappa$ is zero there is no dependence between stock prices and credit events. In this scenario, the calculation of vulnerable valuations is more straightforward because market and credit variables can be treated independently \citep{Gregory2020}.
\end{remark}

\subsection{The forward contract as a portfolio of European options}
\label{sec:general-k}
We first introduce the Black-Scholes pricing function for a European call option $\mathcal{BS}_C:\mathbb{R}_{+}^{*}\times(0,\bar{T}]\rightarrow\mathbb{R}_{+}^{*}$, where we recall this claim has payoff function $f_C(s):=(s-K)^+$ for a strike $K$. 
\newline

Formally, given a log-normal random variable $X:\Omega\rightarrow\mathbb{R}$ with mean $x$ and log-variance $V(\ln X)=\sigma^2y$, we define the Black-Scholes call function for $(x,y)\in\mathbb{R}_{+}^{*}\times(0,\bar{T}]$ as follows:
	\begin{align}
		\label{eq:def-bs-call-expectation}
		\mathcal{BS}_C(x,y)
		:=\mathbb{E}_x((X-K)^+)
	\end{align}	
We state the following, well-known theorem without proof.

\begin{theorem}
	\label{th:black-scholes}
	The Black-Scholes call function satisfies:
		\begin{align}
			\label{eq:bs-call-formula}
			\mathcal{BS}_C(x,y)
			=x\Phi(d_1(x,y))-K\Phi(d_2(x,y))
		\end{align}	
	where $\Phi:\mathbb{R}\rightarrow[0,1]$ is the Gaussian integral defined as:
		\begin{align}
			\label{eq:gaussian-cdf}
			\Phi(z):=
			\int_{-\infty}^z\frac{1}{\sqrt{2\pi}}
			e^{-\frac{\omega^2}{2}}\diff\omega,
			\quad z\in\mathbb{R}
		\end{align}
	and we introduce the functions $d_1,d_2:(0,\bar{T}]\rightarrow\mathbb{R}$:
		\begin{align}
			\label{eq:def-d1}
			d_1(x,y)&:=\frac{1}{\sigma\sqrt{y}}\left(
			\ln\left(\frac{x}{K}\right)+\frac{1}{2}\sigma^2y
			\right)
			\\[6pt]\label{eq:def-d2}
			d_2(x,y)&:=d_1(x,y)-\sigma\sqrt{y}
		\end{align}	
\end{theorem}
\begin{proof}
	The formula was originally derived in \citet{BlackScholes1973}; we refer to Theorem 2.3.2.1 in \citet{JeanblancYorChesney2009} for a proof.
\end{proof}

We define equivalently the Black-Scholes function $\mathcal{BS}_P$ for the price of a European put option with terminal payoff $f_P(s):=(K-s)^+$:
\begin{align}
	\label{eq:bs-put-formula}
	\mathcal{BS}_P(x,y):=
	K\Phi(-d_2(x,y))-x\Phi(-d_1(x,y))
\end{align}

We now state a valuation theorem for a forward contract, valid for any $K\in\mathbb{R}_+^*$ and which gives an analytical representation of its pre-default valuation function in terms of European option values.

\begin{proposition}
	\label{prop:valuation-forward}
	Under the assumptions of Theorem \ref{th:conditional-expectation} as well as that of risk-free close-out, the pre-default value of a vulnerable forward contract with terminal payoff function \eqref{eq:forward-payoff} is given for any $t\in[0,T]$ by:
		\begin{align}
			\notag
			v(t,s)&=
			e^{-r_V(T-t)}
			\left(F_T(t,s)
			e^{\kappa(h_S-r_V)(T-t)}-K
			\right)
			\\[2pt]\notag&\qquad
			+\varrho_1\int_t^T
			e^{-r_V(u-t)}e^{-r(T-u)}
			\mathcal{BS}_C\left((1+\kappa)F_T(t,s)
			e^{\kappa(h_S-r_V)(u-t)},u-t
			\right)\diff u
			\\[3pt]\label{eq:valuation-forward}
			&\qquad
			-\varrho_2\int_t^T
			e^{-r_V(u-t)}e^{-r(T-u)}
			\mathcal{BS}_P\left((1+\kappa)F_T(t,s)
			e^{\kappa(h_S-r_V)(u-t)},u-t
			\right)\diff u
		\end{align}
\end{proposition}
\begin{proof}
	See Appendix \ref{app:proof-th-valuation-forward}. Notice each term in $s$ inside the representation \eqref{eq:valuation-forward} can be bounded from above by a linear term in $s$, including the integrand terms given $0\leq\Phi(x)\leq1$ for any real $x$. Therefore we can find a polynomial in $s$ which bounds $|v(t,s)|$ from above, and the same is true for the MTM function $m(t,s)$. Thus the polynomial growth conditions \eqref{eq:polynomial-growth-f}, \eqref{eq:polynomial-growth-m} and \eqref{eq:polynomial-growth-v} are satisfied, and Theorem \ref{th:conditional-expectation} is applicable.
\end{proof}

Representation \eqref{eq:valuation-forward} corresponds to a linear combination between a forward, and two continuous portfolios of European call and put options. Hence the pricing formula demonstrates the forward can be replicated using a portfolio of options with carefully chosen strikes and expiries, valued under the probability measure $\widehat{\mathbb{Q}}$.
\newline

The first term in \eqref{eq:valuation-forward} captures the value arising from the contractual payoff at expiry $T$ if no counterparty has defaulted beforehand: we refer to it as the \emph{terminal} value component. The functional form of this term is similar to that of the risk-free valuation \eqref{eq:risk-free-forward} but is also sensitive to credit parameters such as credit spreads $\gamma_1,\gamma_2$ and the jump size $\kappa$. 
\newline 

The two integrals in \eqref{eq:valuation-forward} capture three additional components: 
	\begin{itemize}
		\item The first integral relates to (1) the recovery amount if the dealer is in credit at default that is $\wthalf{M}_\tau>0$, and (2) the incremental funding cost when the cash balance is negative due to asymmetric rates.
		\item The second integral depends on the recovery amount in case the dealer is liable at $\tau$.
	\end{itemize}
To align with the literature we shall refer to each integral as the \emph{credit adjustment} (or CVA) and \emph{debit adjustment} (DVA) components respectively%
	\footnote{Our definition for the debit component captures what is usually known as DVA as well as the funding valuation adjustment (FVA) arising from the funding spread $\varphi$. For a discussion on the entanglement between DVA and FVA (or lack thereof) we refer to \citet{HullWhite2012}.}. 
Because of the optionality embedded in these recovery terms, we view the terminal component as the \emph{intrinsic value} of the vulnerable forward: the value of the cash flow to be settled at expiry if default does not materialise.
\newline

Formula \eqref{eq:valuation-forward} is reminiscent of the formula for CVA under wrong way risk derived by \citet{MercurioLi2015} for FX European-style claims. However our model extends their results by considering bilateral wrong way risk and asymmetric funding; moreover, while their results rely on a suitable approximation to \eqref{eq:valuation-forward}, in the next section we show our formula admits an analytical expression under very mild conditions.

\subsection{Analytical valuation formula for the vulnerable forward contract}

Let us define the \emph{normalized log-moneyness} $\eta:\mathbb{R}_+^*\rightarrow\mathbb{R}$ for the vulnerable forward claim $\widehat{\mathcal{C}}_1$:
\begin{align}
	\label{eq:moneyness}
	\eta(s):=
	\frac{1}{\sigma}\ln\frac{(1+\kappa)F_T(t,s)}{K}
\end{align}
Function $\eta$ measures the distance between the contract strike and forward price, adjusted by the jump size $\kappa$. Before stating our main result, we define the following variables to lighten notation:
\begin{align}
	\lambda^*_V:=r_V-r,\qquad
	\zeta_1:=\frac{2\kappa(h_S-r_V)+\sigma^2}{2\sigma},\qquad
	\zeta_2:=\zeta_1-\sigma
\end{align}

The following theorem gives an exact analytical formula, involving only elementary functions and the Gaussian integral \eqref{eq:gaussian-cdf}, for the value of a vulnerable forward contract.

\begin{theorem}
	\label{th:analytical-fwd-value}
	Under the assumptions of Proposition \ref{prop:valuation-forward}, for any $t\in[0,T]$ and letting $F:=F_T(t,s)$, the pre-default value of a vulnerable forward contract with terminal payoff function \eqref{eq:forward-payoff} is equal to:
	\begin{align}
		\notag
		&v(t,s)=
		e^{-r(T-t)}\biggl(e^{-(r_V-r)(T-t)}\left(F
		e^{\kappa(h_S-r_V)(T-t)}-K
		\right)\\[3pt]\notag
		& +\varrho_1\left(
		(1+\kappa)F
		\Lambda\left(T-t,(\lambda^*_V-\kappa(h_S-r_V)),
		\zeta_1,\eta(s)\right)
		-K\Lambda\left(T-t,\lambda^*_V,
		\zeta_2,\eta(s)\right)\right)
		\\[3pt]\label{eq:analytical-fwd-value}
		& -\varrho_2\left(
		K\Lambda\left(T-t,\lambda^*_V,
		-\zeta_2,-\eta(s)\right)
		-(1+\kappa)F
		\Lambda\left(T-t,(\lambda^*_V-\kappa(h_S-r_V)),
		-\zeta_1,-\eta(s)\right)
		\right)\biggr)
	\end{align}
	where $\Lambda:\mathbb{R}_+\times\mathbb{R}^3\rightarrow\mathbb{R}_+$ is the function defined as follows for $x\in\mathbb{R}^*$ and $\rho(x,y)\neq0$%
		\footnote{Analytical formul\ae\ for $\Lambda$ are also derived in Appendix \ref{app:proof-th-analytical-fwd-value} for the particular cases where $\rho(x,y)=0$ (see Corollary \ref{coro:Lambda-zero-rho}); $x=0$ with $y\neq0$ (see Corollary \ref{coro:Lambda-zero-x}); and $x=y=0$ (see Corollary \ref{coro:Lambda-zero-xy}).}:
	\begin{align}
		\notag
		&\Lambda(t,x,y,z)=\\[4pt]
		\label{eq:Lambda-rho-nonzero}
		&\
		\begin{dcases}
			\frac{1}{x}\biggl(
			-e^{-xt}\Phi\left(\beta_0\right)
			-\frac{e^{-yz}}{2}\left(
			e^{-z\rho}\left(\frac{y}{\rho}-1\right)
			\Phi\left(-\beta_1\right)
			-e^{z\rho}\left(\frac{y}{\rho}+1\right)
			\Phi\left(\beta_2\right)
			\right)\biggr),
			& z<0 
			\\[4pt]
			\frac{1}{x}\biggl(
			1-e^{-xt}\Phi\left(\beta_0\right)
			+\frac{e^{-yz}}{2}\left(
			e^{-z\rho}\left(\frac{y}{\rho}-1\right)
			\Phi\left(\beta_1\right)
			-e^{z\rho}\left(\frac{y}{\rho}+1\right)
			\Phi\left(-\beta_2\right)
			\right)\biggr), 
			& z\geq0,
		\end{dcases}
	\end{align}
	where $\rho:=\rho(x,y)$ is the complex-valued function defined as:
	\begin{align}
		\rho(x,y):=
		\begin{dcases}
			\sqrt{2x+y^2}, 		& 2x+y^2\geq0\\
			i\sqrt{|2x+y^2|},	& 2x+y^2<0
		\end{dcases}
	\end{align}
	and we have inserted the following variables:
	\begin{align}
		\beta_0:=\frac{yt+z}{\sqrt{t}},\qquad
		\beta_1:=\frac{\rho(x,y)t-z}{\sqrt{t}},\qquad
		\beta_2:=\beta_1+\frac{2z}{\sqrt{t}}
	\end{align}
\end{theorem}
\begin{proof}
	See Appendix \ref{app:proof-th-analytical-fwd-value}.
\end{proof}

Function $\Lambda$ arises as a closed-form formula to the following integral:
\begin{align}
	\label{eq:Lambda}
	\Lambda(t,x,y,z)=
	\int_0^te^{-xu}\Phi\left(y\sqrt{u}+\frac{z}{\sqrt{u}}\right)\diff u
\end{align}
The previous integral originates in the Black-Scholes integrands from \eqref{eq:valuation-forward}. It is interpreted as a weighted probability for the forward contract to be in-the-money (ITM): we recall that in the Black-Scholes formulas \eqref{eq:bs-call-formula} and \eqref{eq:bs-put-formula}, the terms in $\Phi$ correspond to the probabilities for the option to expire ITM under two specific measures \citep[see Theorem 2 in][]{GemanKarouiRochet1995}.

\begin{remark}
	Although the expression for $\Lambda$ involves complex numbers, the function is actually real-valued. To see this, it suffices to observe that the integral in \eqref{eq:Lambda} is defined with respect to real-valued integrands, therefore it must be real-valued itself.
\end{remark}

The next remark quantifies the benefits from the analytical formula \eqref{eq:analytical-fwd-value} with respect to the option representation \eqref{eq:valuation-forward}.

\begin{remark}
	We integrate numerically \eqref{eq:valuation-forward} and we calculate \eqref{eq:analytical-fwd-value} using an implementation of the Gaussian CDF from the scientific library \emph{\texttt{scipy}} in Python; we find the difference between the two values to be near machine precision for 64-bit float numbers (i.e. $2.2\times10^{-16}$). We measure execution time for both calculations using module \emph{\texttt{timeit}} and find the analytical expression \eqref{eq:analytical-fwd-value} to be around \textbf{2 orders of magnitude} ($10^2$) quicker to compute than integrating \eqref{eq:valuation-forward}.
\end{remark}

When aggregated across thousands of trades in an industrial environment with competing requests, using the closed-form formula \eqref{eq:analytical-fwd-value} allows for material time efficiencies and optimization of computing resources. The formula can be used for obtaining quick price estimates for vulnerable forward contracts traded with risky counterparties, allowing trading desks to quote all-inclusive prices to clients without resorting to Monte Carlo methods. The formula can be integrated into Equity or FX electronic trading systems, enabling dealers to immediately execute requests for quotes (RFQ) from clients while ensuring adequate risk management of funding and credit risk factors.
\newline

To conclude our theoretical findings, we express the results of Theorem \ref{th:analytical-fwd-value} in terms of our original problem.

\begin{corollary}
	The valuation problem \ref{problem} admits an analytical expression \eqref{eq:gaussian-cdf} involving only elementary functions and the Gaussian integral.
\end{corollary}

\section{Numerical analysis of the impact of funding, credit and wrong way risk factors on forward valuation}
\label{sec:numerical-analysis}

We use our valuation model for vulnerable claims to perform numerical analysis on credit and funding factors. We address two questions: first, how do valuations and exposures of a vulnerable claim differ from those of a claim with no credit or wrong way risks? Second, what is the sensitivity of the contract's valuation to different model parameters?

\subsection{Time-zero valuation benchmarking against risk-free contracts}
\label{sec:benchmarking-numerical-analysis}
We consider a numerical example for an equity forward with 5-year expiry ($T=5$) and compare time-zero valuations for a risk-free contract $\mathcal{C}_1^*(f)$ against that of the vulnerable contract $\widehat{\mathcal{C}}_1(f,z_1,z_2)$ valued using the analytical expression from equation \eqref{eq:analytical-fwd-value}.
\newline

Hereafter we assume the risk-free rate $r$ from equation \eqref{eq:risk-free-forward} is equal to the deposit rate $r_\ell$:
\begin{align}
	\label{eq:rf-rate-equal-deposit-rate}
	r:=r_\ell
\end{align}
This is a standard assumption in the literature, see \citet{BrigoMorini2011} for example. 
\newline

Without loss of generality we set $t=0$, $s=1$ and $\alpha=0.5$. The value for $s$ is equivalent to studying the following terminal payoff $f$:
\begin{align}
	\label{eq:trs-payoff}
	f\left(\frac{S_T}{S_0}\right)=\frac{S_T}{S_0}-K,
\end{align}
which corresponds to an \emph{equity return forward}. Such instruments, together with \emph{equity return swaps}, have become increasingly popular among leveraged counterparties such as hedge funds: regulatory financial statements (form FR Y-9C) reported by the five largest US broker-dealers%
\footnote{Bank of America, Citigroup, Goldman Sachs, JP Morgan Chase, Morgan Stanley.}
show the total notional of instruments booked as equity swaps/forwards increased from \$1.3Tn in March 2016 to \$2.6Tn in September 2023 -- a two-fold increase.
\newline

We are interested in the contributions from credit and wrong way parameters, therefore we assume all interest rates are equal and the stock does not pay dividends. These static values are summarised in Table \ref{tab:static-params}.%
\begin{table}[h!]
	\centering
	\begin{threeparttable}
	\begin{tabularx}{0.725\textwidth}
		{>{\raggedright\arraybackslash}m{2.35cm}
		*8{>{\centering\arraybackslash}X}}
		\toprule
		\textbf{Parameter} & $r$ &$r_\ell$ &$\varphi$ &$h_S$ &$h_1$ &$h_2$ &$\sigma$ &$q$
		\\\cmidrule(lr){1-9}
		\textbf{Base value} &0.04 &0.04 &0 &0.04 &0.04 &0.04 &0.3 &0
		\\\bottomrule
	\end{tabularx}
	\caption{Values for static parameters in benchmarking analysis against risk-free contract.}
	\label{tab:static-params}
	\end{threeparttable}
\end{table}%

In order to develop intuition, consider the case represented by Table \ref{tab:static-params} where all market rates are equal to $r$. Let $\gamma_1,\gamma_2=\gamma>0$ and $\varkappa_1,\varkappa_2=\varkappa\in(0,1]$, then the vulnerable forward's price admits the approximation in $\gamma T$:
	\begin{align}
		\label{eq:approx-vulnerable-value}
		v(0,S_0)&=
		v^*(0,S_0)\left(1-\gamma T(1-\varkappa)\right)-\gamma T(1-\varkappa)\kappa S_0+\mathcal{O}(\gamma^2T^2)
	\end{align}
Approximation \eqref{eq:approx-vulnerable-value} is an average of two terms, the risk-free value $v^*(0,S_0)$ and the expected price jump $\kappa S_0$ at default, weighted by the expected loss per unit of notional, $\gamma T(1-\varkappa)$. The formula shows the magnitude of the valuation adjustment linked to wrong way risk is proportional to the probability of default $\gamma T$ and the loss rate $(1-\varkappa)$.
\newline

We analyse numerically the \emph{valuation spread} between vulnerable and risk-free versions of an ATM contract with strike $K=F$, which in this case is equal to:
	\begin{align}
		v(0,s)-v^*(0,s)=v(0,s)
	\end{align}
because $v^*(0,S_0)=0$ when $K=F$. We study different value combinations for credit spreads $(\gamma_1,\gamma_2)$ and jump size $\kappa$, under four distinct scenarios for recovery and creditworthiness which we summarise in Table \ref{tab:benchmark-analysis-scenarios}.%
\newline

\begin{table}[H]
	\centering
	\begin{tabular}{|c|c|}
		\hline
		\parbox{5.2cm}{\vspace{.2cm}%
			{\centering\textbf{Scenario (a)}\par}
			High recovery: $\varkappa_1,\varkappa_2=0.9$\\
			Symmetric credit: $\gamma_1=\gamma_2$%
		\vspace{.2cm}}
		& \parbox{5.2cm}{\vspace{.2cm}%
			{\centering\textbf{Scenario (b)}\par}
			Low recovery: $\varkappa_1,\varkappa_2=0.6$\\ 
			Symmetric credit: $\gamma_1=\gamma_2$%
		\vspace{.2cm}}
		\\\hline
		\parbox{5.2cm}{\vspace{.2cm}%
			{\centering\textbf{Scenario (c)}\par}
			High recovery: $\varkappa_2=0.9$\\ 
			\mbox{Riskless dealer: $\gamma_1=0$, $\varkappa_1=1$}%
		\vspace{.2cm}}
		& \parbox{5.2cm}{\vspace{.2cm}%
			{\centering\textbf{Scenario (d)}\par}
			Low recovery: $\varkappa_2=0.6$\\ 
			\mbox{Riskless dealer: $\gamma_1=0$, $\varkappa_1=1$}%
		\vspace{.2cm}}
		\\\hline
	\end{tabular}
	\caption{Scenarios for benchmarking against risk-free valuation; numerical results are displayed in Table \ref{tab:benchmark-analysis}.}
	\label{tab:benchmark-analysis-scenarios}
\end{table}

Numerical results under all four scenarios are displayed in Table \ref{tab:benchmark-analysis} where the valuation spread is expressed in basis points (bps).
\newline

\setlength{\fboxrule}{.75pt}
\begin{table}[h!]
	\begin{subtable}{.5\textwidth}
		\centering
		\rowcolors{3}{}{gray!10}
		\begin{tabular}{lccccc}
			\toprule
			\multirow{2}{*}{\bm{$\kappa$}} 
			& \multicolumn{5}{c}{\bm{$\gamma_1=\gamma_2$}}	\\\cmidrule(lr){2-6}
			~    		& \it 0 & \it 0.01 	& \it 0.02  & \it 0.03  & \it 0.05 \\\midrule
			\it 0 		& 0.0 	& 0.0 		& 0.0 		& 0.0 		& 0.0 	\\ 
			\it -0.05	& 0.0 	& 2.3 		& 4.1 		& 5.6 		& 7.7	\\
			\it -0.1 	& 0.0 	& 4.5 		& 8.3 		& 11.3 		& 15.5 	\\
			\it -0.2 	& 0.0 	& 9.1 		& 16.7  	& 22.9 		& 31.9	\\
			\it -0.3 	& 0.0 	& 13.8 		& 25.3  	& 34.9 		& 49.1	\\
			\bottomrule
		\end{tabular}
		\caption{$\varkappa_1=\varkappa_2=0.9$}
		\label{tab:benchmark-analysis-1}
	\end{subtable}
	\begin{subtable}{.5\textwidth}
	\centering
	\rowcolors{3}{}{gray!10}
	\begin{tabular}{lccccc}
		\toprule
		\multirow{2}{*}{\bm{$\kappa$}} 
		& \multicolumn{5}{c}{\bm{$\gamma_1=\gamma_2$}}	\\\cmidrule(lr){2-6}
		~    		& \it 0 & \it 0.01 	& \it 0.02  & \it 0.03  & \it 0.05 \\\midrule
		\it 0 		& 0.0 	& 0.0		& 0.0 	& 0.0 	& 0.0 	\\ 
		\it -0.05 	& 0.0 	& 9.1		& 16.5 	& 22.4 	& 30.7 	\\
		\it -0.1 	& 0.0 	& 18.2		& 33.1 	& 45.1 	& 62.2 	\\
		\it -0.2 	& 0.0 	& 36.6		& 66.8 	& 91.6  & 127.6 \\
		\it -0.3 	& 0.0 	& 55.1		& 101.3	& \fbox{139.5} & 196.3 \\
		\bottomrule
	\end{tabular}
	\caption{$\varkappa_1=\varkappa_2=0.6$}
	\label{tab:benchmark-analysis-2}
	\end{subtable}
	\\[9pt]
	\begin{subtable}{.5\textwidth}
		\centering
		\rowcolors{3}{}{gray!10}
		\begin{tabular}{lccccc}
			\toprule
			\multirow{2}{*}{\bm{$\kappa$}} 
			& \multicolumn{5}{c}{\bm{$\gamma_2$} ($\gamma_1=0$)}	\\\cmidrule(lr){2-6}
			~    		& \it 0 & \it 0.01 	& \it 0.02  & \it 0.03  & \it 0.05 \\\midrule
			\it 0 		& 0.0 	& -8.6 		& -16.6 	& -24.2 	& -38.1 \\ 
			\it -0.05	& 0.0 	& -7.2 		& -14.1 	& -20.6 	& -32.6	\\
			\it -0.1 	& 0.0 	& -6.0 		& -11.8 	& -17.3 	& -27.7	\\
			\it -0.2 	& 0.0 	& -4.0 		& -8.0  	& -11.8 	& -19.2	\\
			\it -0.3 	& 0.0 	& -2.5 		& -5.0  	& -7.5 		& -12.5	\\
			\bottomrule
		\end{tabular}
		\caption{$\varkappa_2=0.9$ ($\varkappa_1=1$)}
		\label{tab:benchmark-analysis-3}
	\end{subtable}
	\begin{subtable}{.5\textwidth}
		\centering
		\rowcolors{3}{}{gray!10}
		\begin{tabular}{lccccc}
			\toprule
			\multirow{2}{*}{\bm{$\kappa$}} 
			& \multicolumn{5}{c}{\bm{$\gamma_2$} ($\gamma_1=0$)}	\\\cmidrule(lr){2-6}
			~    		& \it 0 & \it 0.01 	& \it 0.02  & \it 0.03  & \it 0.05 \\\midrule
			\it 0 		& 0.0 	& -34.2		& -66.5 	& -96.8 	& -152.2 	\\ 
			\it -0.05 	& 0.0 	& -28.9		& -56.4 	& -82.4 	& -130.5 	\\
			\it -0.1 	& 0.0 	& -24.1		& -47.2 	& -69.4		& -110.7 	\\
			\it -0.2 	& 0.0 	& -16.1		& -31.8 	& -47.2  	& -76.8 \\
			\it -0.3 	& 0.0 	& -10.0		& -19.9		& -29.9 	& -49.9 \\
			\bottomrule
		\end{tabular}
		\caption{$\varkappa_2=0.6$ ($\varkappa_1=1$)}
		\label{tab:benchmark-analysis-4}
	\end{subtable}
	\caption{Valuation spread, in bps, between vulnerable $\widehat{\mathcal{C}}_1$ and risk-free $\mathcal{C}_1^*$ ATM contracts ($K=F$) for different values of hazard rates and jump, under the numerical configuration from Table \ref{tab:static-params}.}
	\label{tab:benchmark-analysis}
\end{table}

In the case where $\kappa=0$, as anticipated by approximation \eqref{eq:approx-vulnerable-value}, in the presence of symmetric credit parameters for both counterparties, an ATM vulnerable contract has roughly same value as its risk-free equivalent because the CVA and DVA recovery components cancel each other -- first row from Subtables \ref{tab:benchmark-analysis-1} and \ref{tab:benchmark-analysis-2}. 
\newline

On the other hand, when the stock price is expected to jump at default, the valuation gap between risk-free and vulnerable claims can widen or tighten depending on the values for the rest of model parameters (13 in total) -- see second row onwards in all four subtables within Table \ref{tab:benchmark-analysis}. In particular, as evidenced when comparing the right-hand side subtables to the left-hand side ones, the magnitude of the adjustment (and thus the valuation gap) increases as the recovery rates $\varkappa_1, \varkappa_2$ decrease, as predicted by the approximation \eqref{eq:approx-vulnerable-value}. Moreover, Subtables \ref{tab:benchmark-analysis-3} and \ref{tab:benchmark-analysis-4} show that even when one counterparty cannot default, wrong way risk still has a material impact on the contract's valuation.
\newline
	
The next example illustrates just how costly neglecting wrong way risk can be for a derivative dealer.

\begin{example}
	\label{ex:no_wwr-vs-wwr}
	Consider an ATM forward contract with \$500m notional. Assume the Credit Default Swap (CDS) market implies a 3\% annual hazard rate for both counterparties. A dealer which neglects wrong way risk will not request any upfront payment from the client. However if the dealer assesses the stock price might fall by 30\% at the default event, according to our model it should request an upfront payment of \textbf{\$7.0m} (\$500m $\times$ 139.5bps from Table \ref{tab:benchmark-analysis-2}).
\end{example}

In conclusion, neglecting wrong way risks leads to significant mispricing of a transaction, in particular when default likelihood and expected loss are material. This last situation can occur when facing highly leveraged counterparties: a recent illustrative example is the failure of Archegos Asset Management, where total losses across its broker-dealers exceeded \$10bn \citep{LewisWalker2021}. A sizeable portion of the hedge fund's exposure was concentrated in five stocks, to the extent that a sharp drop in their share prices (between 30\% to 50\% in three days, which shows our sample values in Table \ref{tab:benchmark-analysis} are realistic) triggered a margin call from its dealers that it could not meet \mycitepalias{see for example the report by}{}{ESMA2022}. For the banks, default from the counterparty coincided with a large jump in the underlying shares prices and therefore in deal valuations, a risk that had not been accounted for \citep{Cesa2023}.

\subsection{Analysis of the impact from wrong way risk on future exposures}

We consider the same 5-year vulnerable equity forward $\widehat{\mathcal{C}}_1$ and work under the assumptions and parameter settings introduced in the previous section, unless otherwise specified. We compare future exposures for this claim with and without jump-at-default from the stock price to quantify the contribution from wrong way risk to risk measures.
\newline

In addition to the initial valuation, broker dealers also estimate supplementary metrics which quantify the credit risk embedded in the transaction: lifelong \emph{Expected Positive Exposure} or EPE (resp. Negative, ENE) as well as \emph{Potential Future Exposure} (PFE). We briefly review these concepts below and refer to \citet{PykhtinZhu2007} for an excellent overview.
\newline

Given a time grid $\mathbb{T}:=\{t_0,t_1,\dots,t_n\}$ starting at deal inception $t_0=t$ and ending by the expiry date $t_n=T$, we define \emph{positive exposure} $E_+:\mathbb{T}\rightarrow\mathbb{R}_+$ (resp. negative, $E_-:\mathbb{T}\rightarrow\mathbb{R}_-$) for the vulnerable contract $\widehat{\mathcal{C}}_1$ as follows for any $t_i\in\mathbb{T}$:
	\begin{align}
		E_+(t_i):=V_{t_i}^+, \qqquad
		E_-(t_i):=-V_{t_i}^-
	\end{align} 
The \emph{EPE profile} is defined as the expectation at time $t_0$ of the exposure vector $(E_+(t_i))_{i=0,\dots,n}$ under some measure $\mathbb{Q}$. An equivalent definition holds for the ENE profile using negative exposure. The \emph{positive PFE profile} (resp. negative) is defined as the element wise high percentile (resp. low) of positive exposure (resp. negative) along $\mathbb{T}$; typical levels for positive and negative PFE are 95\% and 5\% respectively \citep{PykhtinZhu2007}%
	\footnote{Some authors define such profiles using discounted exposures. In our case where interest rates are not correlated to the stock price $S$, such distinction is not important.}.
\newline

An additional output derived from positive PFE is \emph{Peak Exposure} (PE) which is defined as the maximum PFE value along the time grid $\mathbb{T}$. Peak Exposure is used by risk managers to evaluate the incremental credit risk against a counterparty arising from dealing a new contract. Derivative dealers establish exposure limits against individual counterparties in terms of PE; such limits prevent credit risk to become concentrated among a handful of clients, and deals which breach PE limits might be blocked by risk management functions.
\newline

We compare expected and potential exposures for the ATM forward contract introduced above, under the numerical configuration from Table \ref{tab:benchmark-analysis-4}. We choose the credit spreads $\gamma_1,\gamma_2$ to be equal to 0.03 while the recovery rates $\varkappa_1,\varkappa_2$ are both set to 0.6. We estimate potential exposures under two scenarios: no wrong way risk ($\kappa=0$), and with wrong way risk where $\kappa=-0.3$. 
\newline

To calculate exposures at future times $t_1,\dots,t_n$ we simulate the stock price $S$ under the measure $\widehat{\mathbb{Q}}$ introduced in equation \eqref{eq:dynamics-s-pricing-measure} up until deal expiry $T$. We choose a weekly time grid and simulate a total of 10,000 paths with antithetic variates. The results are displayed in Figure \ref{fig:pfe} where exposures are expressed as a percentage of trade notional.

\begin{figure}[H]
	\centering
	\begin{minipage}{\textwidth}
		\centering
		\includegraphics[width=\textwidth]{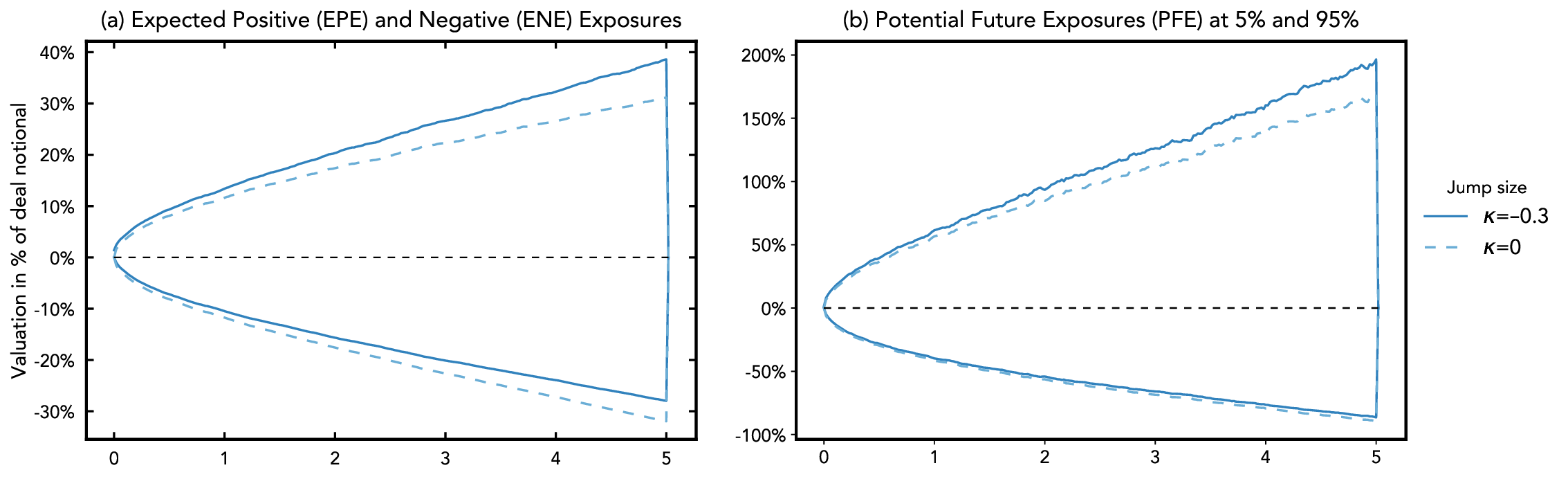}
	\end{minipage}
	\captionsetup{justification=justified}
	\par\vspace{-.2cm}
	\caption{Expected (EPE, ENE) and Potential Future exposures, in percentage of notional, for contract $\widehat{\mathcal{C}}_1$ with and without wrong way risk, under the numerical configuration from Table \ref{tab:static-params}, $\gamma_1=\gamma_2=0.03$ and $\varkappa_1=\varkappa_2=0.6$.}
	\label{fig:pfe}
\end{figure}

We notice immediately that by modelling wrong way risk, exposures are shifted upwards and become higher for the dealer, as the CVA component of the valuation is effectively increased while the DVA one is lowered when $\kappa$ decreases from 0 to -0.3. Dealer PE is materially higher under wrong way risk as seen in Figure \ref{fig:pfe}(b), with the metric increasing from 169\% of notional up to 196\% -- an increase higher than a quarter of trade notional. 

\begin{example}
	Using the same notional amount from Example \ref{ex:no_wwr-vs-wwr}, this corresponds to a \textbf{\$135m} difference in the PE estimate with and without wrong way risk. 
\end{example}

While our estimate assumes neither counterparty posts collateral (which has the effect of decreasing exposure) these results illustrate how neglecting wrong way risk leads to inadequate business decisions: exposure underestimation against a counterparty leads to dangerous credit risk concentration and exposes the dealer to hefty losses in case of counterparty default, as occurred when Archegos Asset Management failed.

\subsection{Multivariate sensitivity analysis}
\label{sec:multivariate-sensitivity-analysis}

We study price sensitivities to various model parameters and identify which ones have the largest impact on valuations, for the same 5-year vulnerable equity forward $\widehat{\mathcal{C}}_1(f,z_1,z_2)$ introduced in Section \ref{sec:benchmarking-numerical-analysis}.
\newline

Our model, while analytical, is also complex with 13 different parameters in total, including the funding policy scalar $\alpha$, which can complicate interpretation. We focus on the sensitivity to credit and wrong way parameters, as well as any interaction between them or with the rest of model parameters. Indeed, credit spreads are recognized as one among two major risk factors to CVA and DVA valuations \mycitepalias{see for example the CVA framework in}{}{BCBS2011} while our analysis in the prior section demonstrates the significance of gap risk at default, therefore it is judicious to target our risk assessment to said parameters.
\newline

We perform our numerical study under assumption \eqref{eq:rf-rate-equal-deposit-rate} for the risk-free rate $r$, namely $r=r_\ell$, and set again $t=0$, $s=1$ and $\alpha=0.5$. We choose to analyse an ATM contract with $K=F$%
	\footnote{When modifying model parameters which affect the ATM level $F$ such as the stock repo rate $h_S$, the contract will no longer be ATM under these different scenarios.} %
to avoid the intrinsic value of the contract dominating its CVA and DVA components.
\newline

For the rest of model parameters, we choose baseline values based on empirical market data (such as S\&P volatilities or US rates) and usual market conventions, which we display in Table \ref{tab:baseline-values}. To aid interpretation and avoid spurious cross effects, all interest rates are set to the same value while the stock is assumed to pay no dividends. When calculating sensitivity to a given parameter and unless otherwise specified, all other parameters are set to the values in this table.

\begin{table}[h!]
	\centering
	\begin{tabularx}{0.9\textwidth}
		{>{\raggedright\arraybackslash}m{2.1cm}
			*{10}{>{\centering\arraybackslash}X}}
		\toprule
		\textbf{Parameter} & $r$ &$r_\ell$ &$\varphi$ &$h_S$ &$h_i$ &$\gamma_i$ &$\varkappa_i$ &$\sigma$ &$q$ &$\kappa$
		\\\cmidrule(lr){1-11}
		\textbf{Base value} &0.04 &0.04 &0 &0.04 &0.04 &0.03 &0.75 &0.3 &0 &0
		\\\bottomrule
	\end{tabularx}
	\caption{Baseline values for model parameters in multivariate sensitivity analysis ($i=1,2$).}
	\label{tab:baseline-values}
\end{table}

Given the large number of parameters and combinations to be analysed, we only highlight the most significant results.

\paragraph{Wrong way risk turbo-charges valuation sensitivity in case of market-wide credit crisis}

Our first observation pertains to the interaction between credit spreads and gap risk. Figure \ref{fig:sensi-g} displays the sensitivity to spreads $\gamma_1,\gamma_2$ under different values for the jump size $\kappa$ with the rest of parameters taking their values from Table \ref{tab:baseline-values}. 
\newline

Subfigures \ref{fig:sensi-g}(a) and \ref{fig:sensi-g}(b) show only one hazard rate varying, a scenario corresponding to an idiosyncratic credit worsening from either counterparty; in Subfigure \ref{fig:sensi-g}(c) both credit spreads widen in unison with their basis $\gamma_2-\gamma_1=1\%$ constant thus representing a scenario of market-wide credit crisis; finally in Subfigure \ref{fig:sensi-g}(d) both spreads also widen but so does their basis, thus capturing a worsening of the client's creditworthiness relative to the dealer.

\begin{figure}[H]
	\centering
	\begin{minipage}{\textwidth}
		\centering
		\includegraphics[width=\textwidth]{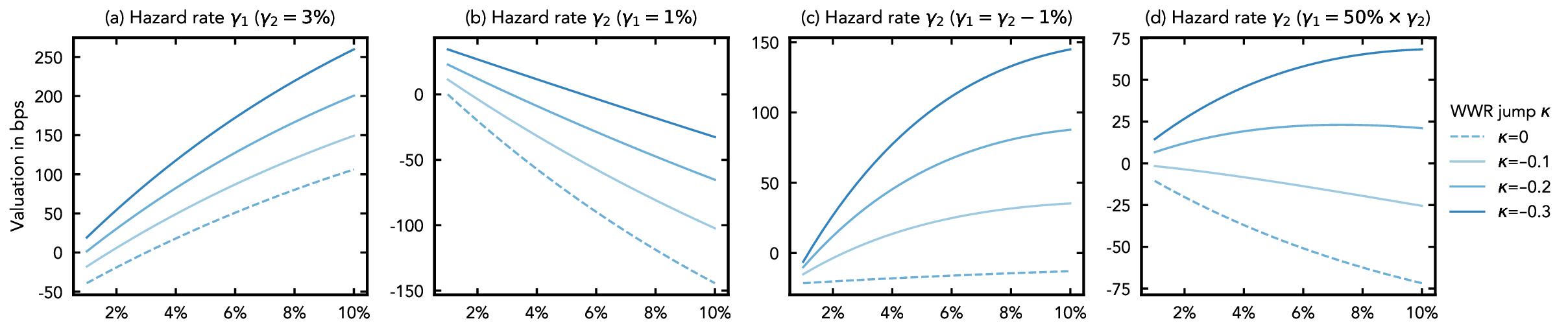}
	\end{minipage}
	\captionsetup{justification=justified}
	\par\vspace{-.2cm}
	\caption{Sensitivity, in bps, to $\gamma_1$ and $\gamma_2$ under different values of $\kappa$.}
	\label{fig:sensi-g}
\end{figure}

Subfigure \ref{fig:sensi-g}(c) proves how the sensitivity to a market-wide credit event (such as the default from a large participant) becomes much sharper in the presence of wrong way risk -- and can even change direction as shown in Subfigure \ref{fig:sensi-g}(d).
\newline

The next example illustrates how material a neglect of wrong way risk can be when managing the contract. 

\begin{example}
	We consider the ATM forward contract with \$500m notional and assume a constant basis $\gamma_2-\gamma_1=1\%$ between the two credit spreads. Consider a scenario where credit spreads jump by 2\%, the client's from from 2\% to 4\% whereas the dealer's from 1\% to 3\%, such that the basis remains unchanged. Based on Subfigure \ref{fig:sensi-g}(c), a dealer which neglects wrong way risk will expect a change in contract value of around 2bps under such a scenario (dashed light blue curve, -20bps against -18bps). However if it assumes a stock price fall of -30\% in case of default, the change in contract value under this scenario will be around 51bps (top dark blue curve, 27bps against 77bps): a difference in expected profit-and-loss of almost \textbf{\$2.5m}.
\end{example}

In conclusion, wrong way not only impacts valuation and exposures of $\widehat{\mathcal{C}}_1$, but also materially widens the valuation sensitivity to credit spreads from both counterparties.

\paragraph{Recovery rate assumptions become critical in the presence of wrong way risk}

Our second observation focuses on recovery rates $\varkappa_1,\varkappa_2$ and the jump size $\kappa$. These parameters are arguably difficult to calibrate: there are no adequate calibration instruments traded in the market with sufficient liquidity, while historical data is scarce. Recovery rates can be estimated based on studies such as \citet{SnP2023} while $\kappa$ can be calibrated to recover stock-credit correlations as suggested by \citet{MercurioLi2015}. In any case, it is important for a dealer to understand how sensitive the valuation is to these choices.
\newline

Figure \ref{fig:sensi-kappa} displays the sensitivity to the jump size $\kappa$ under different values for the recovery rates $\varkappa_1$ and $\varkappa_2$. In the first two Subfigures \ref{fig:sensi-kappa}(a) and \ref{fig:sensi-kappa}(b) we show sensitivities to $\kappa$ against changes in a single recovery rate ($\varkappa_1$ and $\varkappa_2$ respectively) whereas in Subfigure \ref{fig:sensi-kappa}(c) the sensitivity is plotted against simultaneous and identical changes in both the client and dealer recovery rates.%

\begin{figure}[H]
	\centering
	\begin{minipage}{.9\textwidth}
		\centering
		\includegraphics[width=\textwidth]{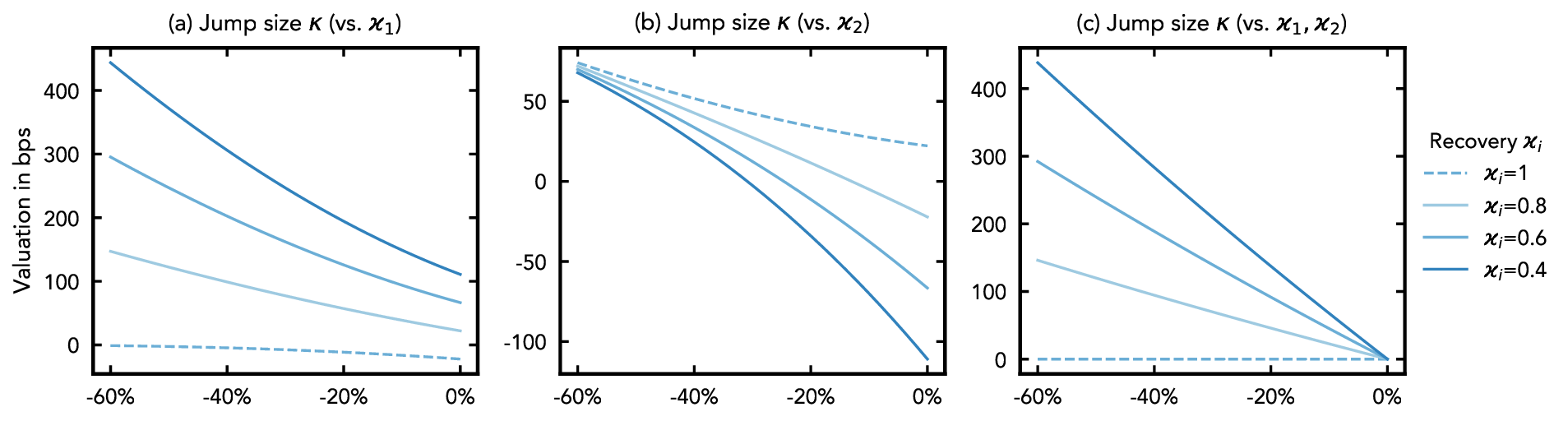}
	\end{minipage}
	\captionsetup{justification=justified}
	\par\vspace{-.2cm}
	\caption{Sensitivity, in bps, to $\kappa$ under different values for $\varkappa_1$ and $\varkappa_2$.}
	\label{fig:sensi-kappa}
\end{figure}

The three subfigures demonstrate the recovery rate values have a critical impact on the valuation sensitivity to changes in the jump size $\kappa$: lower recovery rates or higher jump sizes both increase the expected close-out payment at default thereby inflating the CVA or DVA components, which explains the cross effect evidenced by Figure \ref{fig:sensi-kappa}. The third subfigure \ref{fig:sensi-kappa}(c) is particularly striking: under a common assumption of 60\% recovery (second curve from the top), the valuation difference between assuming either a -20\%  or a -30\% relative jump from $S$ is around 48bps (92bps against 140bps) which corresponds again to around \textbf{\$2.5m} under our \$500m notional assumption.
\newline

A derivative dealer who wishes to capture wrong way risk might be forced to rely on scarce data and expert judgement to choose a value for $\kappa$. Our numerical results illustrate how such assumptions can produce materially different valuations for the same contract depending on the value chosen for recovery rates, and highlight the importance of computing valuations under different ``what-if'' scenarios to properly evaluate the credit risk of a deal.

\paragraph{The recovery rate $\varkappa_2$ shapes the sensitivity to the repo spread $h_2-r_\ell$}

Our final observation is on a significant cross effect between funding and credit parameters, which is not easy to infer from inspecting the option representation \eqref{eq:valuation-forward}. 
\newline

For this analysis, we choose to set $\varphi=3\%$ in order to preserve the no-arbitrage condition \eqref{eq:funding-arbitrage} from Proposition \ref{prop:no-arbitrage}, as well as $\alpha=1$ which implies full repo funding for bond 2 (this is the same setting as in \citeauthor{BurgardKjaer2011}, \citeyear{BurgardKjaer2011}).
\newline
	
Figure \ref{fig:sensi-varkappa2} displays the sensitivity to the spread between bond 2's repo rate $h_2$ and the deposit rate $r_\ell$ under different assumptions for the recovery rate $\varkappa_2$. In Subfigures \ref{fig:sensi-varkappa2}(a) and \ref{fig:sensi-varkappa2}(b) we only vary one of the parameters, while in Subfigure \ref{fig:sensi-varkappa2}(c) we show the sensitivity when we vary both variables simultaneously and identically while keeping the repo spread $h_2-r_\ell$ constant and equal to 1\%.

\begin{figure}[H]
	\centering
	\begin{minipage}{.9\textwidth}
		\centering
		\includegraphics[width=\textwidth]{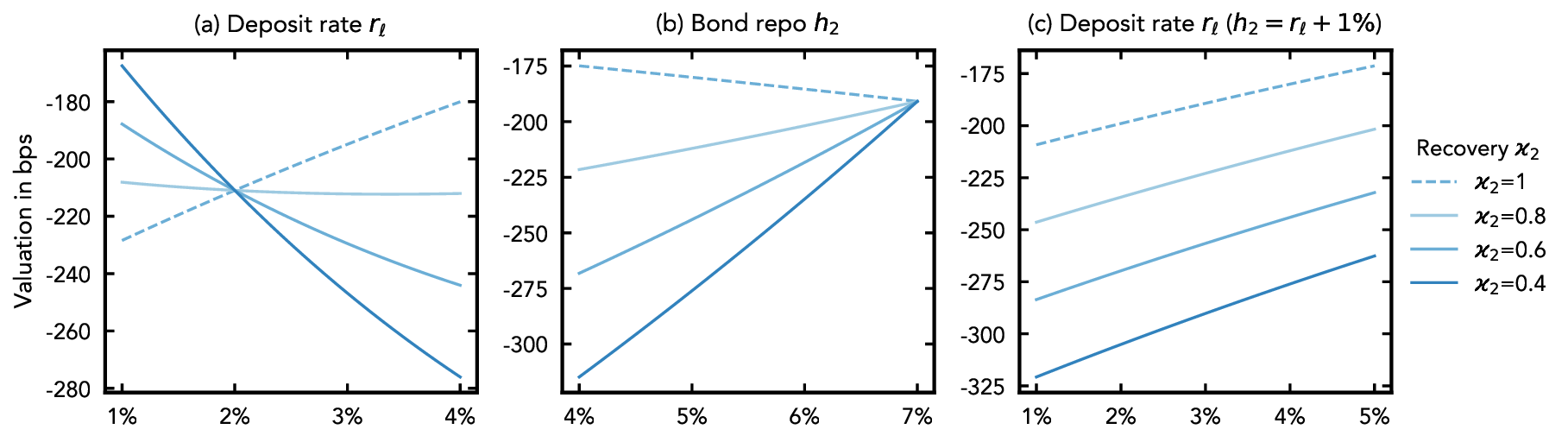}
	\end{minipage}
	\captionsetup{justification=justified}
	\par\vspace{-.2cm}
	\caption{Sensitivity, in bps, to $r_\ell$ and $h_2$ under different values of $\varkappa_2$ with $\alpha=1$ and $\varphi=3\%$.}
	\label{fig:sensi-varkappa2}
\end{figure}

We observe that in the third subfigure \ref{fig:sensi-varkappa2}(c), the slope from the different sensitivity curves remains constant; while in the previous two subfigures the curve's slope rotates as we modify the values from $r_\ell$ and $h_2$ respectively. We conclude there is a cross effect between the recovery rate $\varkappa_2$ and the repo spread $h_2-r_\ell$: by varying either of the two parameters and keeping the other fixed, as done in Subfigures \ref{fig:sensi-varkappa2}(a) and \ref{fig:sensi-varkappa2}(b), the repo spread changes value. 
\newline

These results illustrate how valuation sensitivity to the repo spread $h_2-r_\ell$ changes under different values for recovery rate $\varkappa_2$ -- which corresponds to the recovery in case the dealer's client defaults first. Subfigure \ref{fig:sensi-varkappa2}(a) shows the slope rotates as we decrease the recovery: as observed when analysing the sensitivity to $\kappa$, assumptions on recovery rates require careful consideration and practitioners need be aware of how their hedging strategy can be different depending on their values.

\section{Conclusions}
\label{sec:conclusions}

In this paper we have introduced a valuation framework which incorporates \emph{ab initio} funding, credit and wrong way risks using a mixed Brownian-Poisson model for asset prices and counterparty defaults. We have derived PDEs associated with the valuation of vulnerable claims, including a linearised one from which we have derived an explicit stochastic representation for the pre-default valuation using the Feynman-Kac formula. We have specialized our results to the case of a forward contract on an equity share, and shown that this claim can be represented in terms of European options. Importantly, we have proved the valuation formula for a forward admits an analytical formula, under a funding policy which linearises the PDE, involving only elementary functions and Gaussian integrals.
\newline

Our numerical analysis demonstrates wrong way risk is a significant risk factor: neglecting price jumps at default leads to severe errors when calculating valuations and future exposures, which can quickly run into the millions for typical trade notionals even under benign assumptions such as 3-5\% credit spreads and 10-30\% jump sizes; our analysis also show wrong way risk increases the sensitivity to other risk factors, in particular credit spreads from both parties. Additionally, our numerical results highlight that assumptions about recovery rates, which are hard to calibrate or estimate, have a material impact on other parameter sensitivities such as jump size or repo spreads.
\newline

Our research contributes to both academic and industrial aspects of derivative valuation under credit, funding and wrong way risk. We prove it is possible to price vulnerable contracts in closed-form even in the presence of asymmetric funding, wrong way risk and price-credit correlation thereby extending the analytical frameworks developed in the work from \citet{BurgardKjaer2011} on PDEs for contracts with credit and funding adjustments; \citet{MercurioLi2015} on wrong way risk at default; and \citet{BrigoBuescuRutkowski2017} on analytical formulas for valuing vulnerable contracts with bilateral cash flows. For derivative dealers, our formulas offer an efficient methodology for quoting prices and estimating sensitivities inclusive of funding, credit and wrong way risk.  They can be particularly useful in electronic trading platforms where speedy execution is crucial.
\newline

Our results apply immediately to unilateral cash flow contracts such as European options by considering only default from the counterparty selling the option, thereby covering the two most common type of derivative contracts traded in financial markets.
\newline

In summary, the valuation framework we have presented generalises existing approaches from the literature; provides new insights particularly regarding wrong way risk; and is amenable to further refinements in several directions. It can accommodate other asset classes which lend themselves to log-normal modelling such as commodities \citep{GibsonSchwartz1990} and inflation \citep{JarrowYildirim2003}. Alternatively, it can be extended to interest rate products such as swaps or swaptions \mycitepalias{which command some of the largest volumes in traded notional among derivatives, see for example}{}{BIS2023} by using normal dynamics for the asset price $S$. While we have considered a deterministic model for the funding spread or hazard rates, the model can be augmented by inserting stochastic dynamics for these risk factors, which would enable to study the impact from correlation between funding and credit variables. Finally, the present framework can also be extended to incorporate variation margin (VM) as well as initial margin (IM) for OTC deals with collateral: in particular this allows studying how collateral agreements mitigate any wrong way risk arising from the transaction.

\clearpage

\input{paper_20240324_bibliography.bbl}
\clearpage
\appendix


\section{Technical Proofs}

\subsection{Proof of Proposition \ref{prop:no-arbitrage}}
\label{app:proof-no-arbitrage}
\begin{proof}
Proposition \ref{prop:no-arbitrage} states a necessary condition for $\mathfrak{M}$ to be arbitrage-free hence it suffices to prove that if conditions \eqref{eq:funding-arbitrage}-\eqref{eq:credit-arbitrage} are not fulfilled, then there exists arbitrage opportunities. To prove \eqref{eq:funding-arbitrage} is a necessary condition, without loss of generality let us consider $h=h_S$ and a trading strategy $\Theta$ such that $\upvartheta^S_t\geq0$ for all $t$ and:
\begin{align}
	\theta^S_t=-\upvartheta^S_t,\qquad
	\theta^\ell_t=\frac{\upvartheta^S_tS_t}{B^\ell_t}
\end{align}
while the rest of the elements of $\Theta$ are null, then:
\begin{align}
	V^\Theta_0=0,\qquad
	G^\Theta_t=(r_\ell-h_S)\int_0^t\upvartheta^S_uS_u\diff u
\end{align}
which is an arbitrage if $r_\ell>h_S$. To show that $r_b<h_S$ entails arbitrage opportunities, it suffices to consider instead the portfolio with $\upvartheta^S_t\leq0$ and:
\begin{align}
	\theta^S_t=-\upvartheta^S_t,\qquad
	\theta^b_t=\frac{\upvartheta^S_tS_t}{B^b_t}
\end{align}
The same reasoning is applicable to the bond repo rates $h_1$ and $h_2$ thus establishing the necessity of \eqref{eq:funding-arbitrage}. To prove \eqref{eq:credit-arbitrage} is necessary, consider an alternative strategy $\Theta$ with \smash{$\upvartheta^{i}_t\leq0$} and the rest of units being null such that:
\begin{align}
	V^\Theta_0=0,\qquad
	G^\Theta_t=(h_i-r_i)\int_0^tP^{i}_{u-}\diff u+\int_0^tP^{i}_{u-}\diff J^{i}_u
\end{align}
\lineskiplimit=-\maxdimen
If $h_i\geq r_i$ then the strategy receives a continuous positive cash flow stream $P^{i}_{u-}(h_i-r_i)\diff u$ until default $\tau$ where it pays out the positive lump sum $P^{i}_\tau$ which constitutes an arbitrage.
\end{proof}

\subsection{Proof of Proposition \ref{prop:hedging-strategy}}
\label{app:proof-lemma-hedging-strategy}

We first state two technical lemmas.
\begin{lemma}
	\label{lemm:ito-for-fv}
	Let $X:[0,T]\times\Omega\rightarrow\mathbb{R}$ be a semimartingale and $Y:[0,T]\times\Omega\rightarrow\mathbb{R}$ a process of bounded variation. Then the following product rule holds:
	\begin{align}
		X_tY_t=X_0Y_0+\int_0^tY_{u-}\diff X_u+\int_0^tX_u\diff Y_u
	\end{align}
\end{lemma}
\begin{proof}
	See proof for Lemma 1.10 in \citet{BieleckiJeanblancRutkowski2004}.
\end{proof}

\begin{lemma}
	\label{lemm:indicator-equivalency}
	The following equality holds for $i,j\in\{1,2\}$ such that $i\neq j$:
		\begin{align}
			\{\tau_i\leq\tau_j\}\cap\{t\geq\tau\}
			=\{(t\wedge\tau_i)\leq\tau\}\cap\{t\geq\tau_i\}
		\end{align}
\end{lemma}
\begin{proof}
	Note that by construction the default processes are independent and cannot jump at the same time, hence the sets $\{\tau_i<\tau_j\}$ and $\{\tau_i\leq\tau_j\}$ have same measure. Then:
	\begin{align}
		\notag
		\{\tau_i\leq\tau_j\}\cap\{t\geq\tau\}
		&=\{\tau_i\leq\tau_j\}\cap\{t\geq(\tau_i\wedge\tau_j)\}\\[7pt]\notag
		&=\{\tau_i\leq\tau_j\}\cap\{t\geq\tau_i\}\\[7pt]\notag
		&=\{\tau_i\leq\tau\}\cap\{t\geq\tau_i\}\\[7pt]
		&=\{(t\wedge\tau_i)\leq\tau\}\cap\{t\geq\tau_i\}
	\end{align}
	where the last equality stems from the equivalence $t\wedge\tau_i\equiv\tau_i$ over the set $\{t\geq\tau_i\}$.
\end{proof}

\begin{proof}[Proof of Proposition \ref{prop:hedging-strategy}]
To derive the hedging strategy $\Theta$, we need to ensure conditions 
	\eqref{eq:netting-condition},
	\eqref{eq:hedging-condition-price}, and
	\eqref{eq:self-financing-condition-gains}
hold for any $t$ up to expiry $T$:
	\begin{itemize}
		\item We use the bank accounts to enforce the hedging condition \eqref{eq:hedging-condition-price} while enforcing the netting condition \eqref{eq:netting-condition}.
		\item The trading strategy in stock and bonds is derived using the self-financing condition \eqref{eq:self-financing-condition-gains}.
	\end{itemize}

Define the \emph{cash balance} $C^{\ell,b}:[0,T]\times\Omega\rightarrow\mathbb{R}$ from strategy $\Theta$ as follows:
	\begin{align}
		\label{eq:cash-balance}
		C^{\ell,b}_t:=\theta^\ell_tB^\ell_t+\theta^b_tB^b_t
	\end{align}
Because the bank account prices $B^\ell,B^b$ are continuous while the strategies $\theta^\ell,\theta^b$ are $\mathbb{F}$-predictable, the expression for the cash balance must preserve left-continuity. We recall from equation \eqref{eq:portfolio-price} that the value of any portfolio in $\mathfrak{M}$ is:
	\begin{align}
		V^\Theta_t:=
		\theta^S_tS_t
		+\sum_{i=1}^2\theta^{i}_tP^{i}_t
		+\theta^\ell_tB^\ell_t
		+\theta^b_tB^b_t
		=\theta^S_tS_t
		+\sum_{i=1}^2\theta^{i}_tP^{i}_t
		+C_t^{\ell_b}
	\end{align}
Hence to ensure $\Theta$ satisfies $V^\Theta+\widehat{V}=0$ and that the cash balance remains $\mathbb{F}$-predictable, $C^{\ell,b}$ must be equal to:
	\begin{align}
		\notag
		C^{\ell,b}_t
		&=-\widehat{V}_{t-}-\theta^S_tS_{t-}-\sum_{i=1}^2\theta^{i}_tP^{i}_{t-}\\\notag
		&=-\indicator{t^-<\tau}V_{t-}-\theta^S_tS_{t-}-\sum_{i=1}^2\theta^{i}_tP^{i}_{t-}\\
		\label{eq1:proof-lemma-hedging-strategy}
		&=-\indicator{t\leq\tau}V_{t-}-\theta^S_tS_{t-}-\sum_{i=1}^2\theta^{i}_tP^{i}_{t-}
	\end{align}
By combining the definition of the cash balance \eqref{eq:cash-balance} with the netting condition $\theta^\ell\theta^b\equiv0$, we obtain the strategies in the deposit and funding accounts respectively:
	\begin{align}
		\label{eq2:proof-lemma-hedging-strategy}
		\theta^\ell_t=\frac{1}{B^\ell_t}(C^{\ell,b}_t)^+,
		\qqquad
		\theta^b_t=-\frac{1}{B^b_t}(C^{\ell,b}_t)^-
	\end{align}
which concludes the first part of the proof.
\newline

To derive the strategies in the stock $(\theta^S,\upvartheta^S)$ and the bonds $(\theta^{i},\upvartheta^{i})$ we start by applying the product rule from Lemma \ref{lemm:ito-for-fv} to the price component $\widehat{V}=(1-J)V$ within the gain process $\widehat{G}$. The Lemma is applicable because $V$ is a function of the bounded variation process $t$ and a semimartingale, thus is semimartingale itself, while $J$ is evidently of bounded variation:
	\begin{align}
		\label{eq3:proof-lemma-hedging-strategy}
		\diff \widehat{G}_t
		=\diff\widehat{V}_t+\diff\widehat{D}_t
		=\left((1-J_{t-})\diff V_t-V_t\diff J_t\right)
		+\sum_{i=1}^2(1-J_{t-})Z^i_t\diff J^{i}_t
	\end{align}
We seek to factorize the previous equation by $(1-J_{t-})=\indicatorf{[0,\tau]}{t}$. Using the definition of jump integrals and applying Lemma \ref{lemm:indicator-equivalency}, note that the second integral in \eqref{eq3:proof-lemma-hedging-strategy} can be written as:
	\begin{align}
		\notag
		\int_0^tV_u\diff J_u
		=V_{\tau}\indicator{t\geq\tau}
		=V_{t\wedge\tau}\indicator{t\geq\tau}
		&=\left(V_{t\wedge\tau_1}\indicator{\tau_1\leq\tau_2}+
		V_{t\wedge\tau_2}\indicator{\tau_1>\tau_2}\right)\indicator{t\geq\tau}\\\notag
		&=\left(V_{t\wedge\tau_1}\indicator{(t\wedge\tau_1)\leq\tau}\right)\indicator{t\geq\tau_1}+
		\left(V_{t\wedge\tau_2}\indicator{(t\wedge\tau_2)\leq\tau}\right)\indicator{t\geq\tau_2}\\[4pt]\notag
		&=\left(V_{\tau_1}\left(1-J_{\tau_1-}\right)\right)\indicator{t\geq\tau_1}+
		\left(V_{\tau_2}\left(1-J_{\tau_2-}\right)\right)\indicator{t\geq\tau_2}\\
		\label{eq4:proof-lemma-hedging-strategy}
		&=\int_0^t(1-J_{u-})V_u\diff J^{1}_u+\int_0^t(1-J_{u-})V_u\diff J^{2}_u
	\end{align}
Injecting equation \eqref{eq4:proof-lemma-hedging-strategy} into \eqref{eq3:proof-lemma-hedging-strategy}, the gain process $\widehat{G}$ is equal to:
	\begin{align}
		\label{eq5:proof-lemma-hedging-strategy}
		\diff \widehat{G}_t
		&=(1-J_{t-})\left(\diff V_t+\sum_{i=1}^2\left(Z^i_t-V_t\right)\diff J^{i}_t\right)
	\end{align}
To expand the term in $\diff V$ above, we apply It\^o's Formula for semimartingales \citep[Theorem 32 from Chapter II in][]{Protter1990}:
	\begin{align}
		\label{eq6:proof-lemma-hedging-strategy}
		\diff V_t=
		\mathscr{L}_Sv(t,S_{t-})\diff t
		+v_s(t,S_{t-})\diff S_t
		+(\Delta V_t-v_s(t,S_{t-})\Delta S_t)(\diff J^{1}_t+\diff J^{2}_t)
	\end{align}
where we define the partial derivative of $v$; the pricing operator $\mathscr{L}_S$ for process $S$; and the jump sizes $\Delta S$ and $\Delta V$, respectively:
\begin{align}
	\label{eq7:proof-lemma-hedging-strategy}
	v_s:=\frac{\partial v}{\partial s},\qquad
	\mathscr{L}_S:=\frac{\partial}{\partial t}+\frac{\sigma^2s^2}{2}\frac{\partial^2}{\partial s^2},\qquad
	\Delta S_t:=S_t-S_{t-},\quad
	\Delta V_t:=V_t-V_{t-}
\end{align}
Injecting the stochastic differential equation \eqref{eq6:proof-lemma-hedging-strategy} satisfied by $V$ into equation \eqref{eq5:proof-lemma-hedging-strategy} and using the definition of jump $\Delta V$, we obtain an integral representation for the gain process $\widehat{G}$:
	\begin{align}
		\label{eq8:proof-lemma-hedging-strategy}
		\diff \widehat{G}_t
		&=(1-J_{t-})\left(\mathscr{L}_Sv(t,S_{t-})\diff t+v_s(t,S_{t-})\diff S_t
		+\sum_{i=1}^2\left(Z^i_t-V_{t-}-v_s(t,S_{t-})\Delta S_t\right)\diff J^{i}_t\right)
	\end{align}
To finalize the proof, we consider the gains from strategy $\Theta$. Replacing the repo dividend processes defined in equations \eqref{eq:stock-repo-model} and \eqref{eq:zcbond-repo-model} into expression \eqref{eq:portfolio-gain-process} for $G^\Theta$; as well using the definition for stock and bond exposures from equation \eqref{eq:def-exposures}, we obtain:
	\begin{align}
		\label{eq9:proof-lemma-hedging-strategy}
		\diff G^\Theta_t&=
		\xi^S_t\diff(S_t+D_t)
		+\sum_{i=1}^2\xi^{i}_t\diff P^{i}_t
		-\left(h_S\upvartheta^S_tS_{t-}
		+\sum_{i=1}^2h_i\upvartheta^{i}_tP^{i}_{t-}\right)
		\diff t
		+\theta^\ell_t\diff B^\ell_t
		+\theta^b_t\diff B^b_t
	\end{align}
The equality $\widehat{G}_t+G^\Theta_t=0$ can only hold if the stochastic components from both gain processes are a.s. opposites on every path. By expanding the terms $\diff P^{i}_t$ using their definition \eqref{eq:zcbond-model}, then matching the stochastic terms from the portfolio gains $G^\Theta$ in \eqref{eq9:proof-lemma-hedging-strategy} to those $\widehat{G}$ from the vulnerable claim in \eqref{eq8:proof-lemma-hedging-strategy}, we obtain the following system:
	\begin{subequations}
	\label{eq10:proof-lemma-hedging-strategy}
	\begin{align}
		&\xi^S_t\diff S_t=-(1-J_{t-})v_s(t,S_{t-})\diff S_t\\[3pt]
		&\xi^{i}_tP^{i}_{t-}\diff J^{i}_t=(1-J_{t-})\left(Z^i_t-V_{t-}-v_s(t,S_{t-})\Delta S_t\right)\diff J^{i}_t
	\end{align}	
	\end{subequations}
By observing the stock jump size $\Delta S$ is equal to $\kappa S_-$ we conclude the proof.
\end{proof}

\subsection{Proof of Proposition \ref{prop:pricing-pde}}
\label{app:proof-prop-pricing-pde}
\begin{proof}
To derive the PDE \eqref{eq:pricing-pde}, we first show we can restrict our attention to the pre-default set $\{t\leq\tau\}$. Then by applying the self-financing condition \eqref{eq:self-financing-condition-gains}, we replace the numbers of units in each asset by their hedging value as derived in Proposition \ref{prop:hedging-strategy}.
\newline

We start by observing that, similar to the claim's gain process $\widehat{G}$ as displayed in \eqref{eq8:proof-lemma-hedging-strategy}, the gain process $G^\Theta$ for the hedging strategy can be factorized by the pre-default process $(1-J_-)$. For the number of units invested in the physical assets and their repos, this is an immediate consequence from equation \eqref{eq10:proof-lemma-hedging-strategy} on exposures, and definition \eqref{eq:funding-policy-stock} linking exposures and units. This also holds for the bank accounts units $\theta^\ell$ and $\theta^b$. Indeed, the cash balance satisfies:
	\begin{align}
		\notag
		C^{\ell,b}_t
		&=-\indicator{t\leq\tau}V_{t-}-\alpha_S\xi^S_tS_{t-}-\sum_{i=1}^2\alpha_i\xi^{i}_tP^{i}_{t-}\\
		\label{eq1:proof-prop-pricing-pde}
		&=(1-J_{t-})\left(-V_{t-}+\alpha_Sv_s(t,S_{t-})S_{t-}-\sum_{i=1}^2\alpha_i\left(Z^i_t-V_{t-}-v_s(t,S_{t-})\kappa S_{t-}\right)P^{i}_{t-}\right)
	\end{align}
For any process $X\in\{\xi^S,\theta^S,\upvartheta^S,\xi^i,\theta^i,\upvartheta^i,\theta^\ell,\theta^b,C^{\ell,b}\}$ we define its pre-default version $\widetilde{X}$ by keeping the default indicator separate, such that for any $t\in[0,T]$:
	\begin{align}
		\label{eq2:proof-prop-pricing-pde}
		X_t=\indicator{t\leq\tau}\widetilde{X}_t
	\end{align}
Then we write the portfolio gain process \eqref{eq9:proof-lemma-hedging-strategy} using the pre-default versions:
	\begin{align}
		\notag
		\diff G^\Theta_t&=
		(1-J_{t-})\Bigg(
		\widetilde{\xi}^S_t\diff(S_t+D_t)
		+\sum_{i=1}^2\widetilde{\xi}^{i}_t\diff P^{i}_t
		-\left(h_S\widetilde{\upvartheta}^S_tS_{t-}
		+\sum_{i=1}^2h_i\widetilde{\upvartheta}^{i}_tP^{i}_{t-}
		\right)\diff t
		\\\label{eq3:proof-prop-pricing-pde}
		&\qqquad
		+\widetilde{\theta}^\ell_t\diff B^\ell_t
		+\widetilde{\theta}^b_t\diff B^b_t
		\Bigg)
	\end{align}
From expression \eqref{eq8:proof-lemma-hedging-strategy} for the claim gains $\widehat{G}$ combined with equation \eqref{eq3:proof-prop-pricing-pde}, the condition $\widehat{G}+G^\Theta=0$ is satisfied either over the set $\{J_{t-}=1\}=\{t>\tau\}$ which represents the post-default state; or when the following equality holds, where we have 
split the bond price changes $\diff P^{i}_t$ into their return and default components, and used the hedging system \eqref{eq10:proof-lemma-hedging-strategy} to remove any stochastic component:
	\begin{align}
		\label{eq4:proof-prop-pricing-pde}
		&\mathscr{L}_SV_{t-}\diff t
		+\widetilde{\xi}^S_t\diff D_t
		-\left(
		h_S\widetilde{\upvartheta}^S_tS_{t-}
		+\sum_{i=1}^2(h_i\widetilde{\upvartheta}^{i}_t
		-r_i\widetilde{\xi}^{i}_t)
		P^{i}_{t-}
		\right)\diff t
		+\widetilde{\theta}^\ell_t\diff B^\ell_t
		+\widetilde{\theta}^b_t\diff B^b_t=0
	\end{align}	
where by abuse of notation we use $\mathscr{L}_SV_{t-}$ to refer to $\mathscr{L}_Sv(t,S_{t-})$. In the previous equation, we replace the stock dividends $\diff D_t$ by their expression, defined in Section \ref{sec:model-description}, equation \eqref{eq:dividend-model}:
	\begin{align}
		\label{eq5:proof-prop-pricing-pde}
		&\mathscr{L}_SV_{t-}\diff t
		+\left(
		\big(q\widetilde{\xi}^S_t-h_S\widetilde{\upvartheta}^S_t\big)S_{t-}
		+\sum_{i=1}^2\big(r_i\widetilde{\xi}^{i}_t
		-h_i\widetilde{\upvartheta}^{i}_t\big)
		P^{i}_{t-}
		\right)\diff t
		+\widetilde{\theta}^\ell_t\diff B^\ell_t
		+\widetilde{\theta}^b_t\diff B^b_t=0
	\end{align}	
We now insert the expression \eqref{eq2:proof-lemma-hedging-strategy} for the bank account units $\theta^\ell, \theta^b$ and replace the bank account gains $\diff B^\ell_t,\diff B^b_t$ by their expressions from equations \eqref{eq:deposit-model} and \eqref{eq:funding-model}. Given now there are only deterministic terms involved, we factorize and divide by $\diff t$:
	\begin{align}
		\label{eq6:proof-prop-pricing-pde}
		&\mathscr{L}_SV_{t-}
		+\sum_{i=1}^2\big(r_i\widetilde{\xi}^{i}_t
		-h_i\widetilde{\upvartheta}^{i}_t\big)P^{i}_{t-}
		+\big(q\widetilde{\xi}^S_t
		-h_S\widetilde{\upvartheta}^S_t\big)S_{t-}
		+r_\ell\widetilde{\theta}^\ell_tB^\ell_t
		+r_b\widetilde{\theta}^b_tB^b_t=0
	\end{align}	
In the above equation \eqref{eq6:proof-prop-pricing-pde}, we replace the (pre-default) number of units $\widetilde{\upvartheta}^S$ in stock repos by the expression $(1-\alpha_S)\widetilde{\xi}^S$, and the units $\widetilde{\upvartheta}^{i}$ in risky bonds by $(1-\alpha_i)\widetilde{\xi}^{i}$:
	\begin{align}
		\label{eq7:proof-prop-pricing-pde}
		{\hspace{-.2cm}}
		\mathscr{L}_SV_{t-}
		+\sum_{i=1}^2\big(r_i-(1-\alpha_i)h_i\big)
		\widetilde{\xi}^{i}_tP^{i}_{t-}
		+\big(q-(1-\alpha_S)h_S\big)
		\widetilde{\xi}^S_tS_{t-}
		+r_\ell\widetilde{\theta}^\ell_tB^\ell_t
		+r_b\widetilde{\theta}^b_tB^b_t=0
	\end{align}	
We define the \emph{funding spread} $\varphi:=r_b-r_\ell$, which is positive by Proposition \ref{prop:no-arbitrage}.
Then we rearrange the bank account gains in \eqref{eq7:proof-prop-pricing-pde} using the pre-default balance $\widetilde{C}^{\ell,b}$:
\begin{align}
	\label{eq9:proof-prop-pricing-pde}
	r_\ell\widetilde{\theta}^\ell_tB^\ell_t
	+r_b\widetilde{\theta}^b_tB^b_t
	=r_\ell\widetilde{C}^{\ell,b}_t
	-\varphi(\widetilde{C}^{\ell,b}_t)^-
\end{align}
We inject the above decomposition \eqref{eq9:proof-prop-pricing-pde} into equation \eqref{eq7:proof-prop-pricing-pde}, then expand the expression of the cash balance using \eqref{eq1:proof-prop-pricing-pde}, so that equation \eqref{eq7:proof-prop-pricing-pde} becomes:
	\begin{multline}
		-r_\ell V_{t-}
		+\mathscr{L}_SV_{t-}
		+(q-(1-\alpha_S)h_S-\alpha_Sr_\ell)
		\widetilde{\xi}^S_tS_{t-}
		\\\label{eq11:proof-prop-pricing-pde}
		+\sum_{i=1}^2(r_i-(1-\alpha_i)h_i-\alpha_ir_\ell)
		\widetilde{\xi}^{i}_tP^{i}_{t-}
		-\varphi\left(
		-V_{t-}-\alpha_S\widetilde{\xi}^S_tS_{t-}
		-\sum_{i=1}^2\alpha_i\widetilde{\xi}^{i}_tP^{i}_{t-}\right)^-=0
	\end{multline}	
We define the new variables $r_S$ and $\lambda_i$ for $i\in\{1,2\}$:
	\begin{align}
		\label{eq12:proof-prop-pricing-pde}
		r_S&:=\alpha_Sr_\ell+(1-\alpha_S)h_S\\[3pt]
		\label{eq13:proof-prop-pricing-pde}
		\lambda_i&:=\gamma_i-(1-\alpha_i)(h_i-r_\ell)
	\end{align}
where we have used assumption \eqref{eq:bond-return-rate} that is $r_i=r_\ell+\gamma_i$; and we write equation \eqref{eq11:proof-prop-pricing-pde} using our newly defined variables:
	\begin{multline}
		\label{eq14:proof-prop-pricing-pde}
		-r_\ell V_{t-}
		+\mathscr{L}_SV_{t-}
		-(r_S-q)\widetilde{\xi}^S_tS_{t-}
		+\sum_{i=1}^2
		\lambda_i\widetilde{\xi}^{i}_tP^{i}_{t-}\\
		-\varphi\left(
		-V_{t-}-\alpha_S\widetilde{\xi}^S_tS_{t-}
		-\sum_{i=1}^2\alpha_i\widetilde{\xi}^{i}_tP^{i}_{t-}
		\right)^-=0
	\end{multline}	
We replace the bond exposure $\widetilde{\xi}^{i}$ by their hedging values given in \eqref{eq10:proof-lemma-hedging-strategy}, so that: 
	\begin{align}
		\label{eq15:proof-prop-pricing-pde}
		\widetilde{\xi}^{i}_tP^{i}_{t-}=
		Z^{i}_t-V_{t-}+\kappa\widetilde{\xi}^S_tS_{t-}
	\end{align}
and we define the following auxiliary variables:
	\begin{align}
		\label{eq16:proof-prop-pricing-pde}
		\lambda_V:=\lambda_1+\lambda_2,\qqquad
		\alpha_P:=\alpha_1+\alpha_2
	\end{align}
Equation \eqref{eq14:proof-prop-pricing-pde} becomes:
	\begin{multline}
		-(r_\ell+\lambda_V)V_{t-}
		+\mathscr{L}_SV_{t-}
		-(r_S-q-\kappa\lambda_V)\widetilde{\xi}^S_tS_{t-}
		+\sum_{i=1}^2\lambda_iZ^{i}_t
		\\\label{eq17:proof-prop-pricing-pde}
		-\varphi\left(
		-(1-\alpha_P)V_{t-}
		-(\alpha_S+\kappa\alpha_P)\widetilde{\xi}^S_tS_{t-}
		-\sum_{i=1}^2\alpha_iZ^{i}_t
		\right)^-=0
	\end{multline}	
We replace $\widetilde{\xi}^S$ by its value implied by Proposition \ref{prop:hedging-strategy}, equation \eqref{eq:stock-hedging}, in order to obtain the desired PDE displayed in equation \eqref{eq:pricing-pde}.
\newline

Finally, if $\{\tau>T\}$ then the value of the claim at expiry must verify $\widehat{V}_T=V_T=f(S_T)$ which entails the terminal condition $v(T,s)={f}(s)$ thus concluding the proof.
\end{proof}

\subsection{Proof of Corollary \ref{coro:pricing-pde-linearised}}
\label{app:proof-coro-pricing-pde-linearised}
\begin{proof}
Under the funding policy $\upalpha^\star$ we have:
	\begin{alignat}{2}
		\label{eq1:proof-coro-pricing-pde-linearised}
		&\alpha_{P}=1&\quad\Rightarrow\quad&1-\alpha_{P}=0
		\\
		&\alpha^\star_S=-\kappa&\quad\Rightarrow\quad& \alpha_S^\star+\alpha_{P}\kappa=0
	\end{alignat}
Replacing these values in PDE \eqref{eq:pricing-pde} from Proposition \ref{prop:pricing-pde} yields the equation:
	\begin{align}
		&-r_Vv
		+\mathscr{L}_Sv
		+(r_S-q-\kappa\lambda_V)v_ss
		+\sum_{i=1}^2\lambda_iz_i
		-\varphi\left(\sum_{i=1}^2\alpha_{i}z_i\right)^+=0
	\end{align}
	where we have used the equivalency $-(-x)^-=-(x)^+$. Under $\upalpha^\star$ we have $r_S=h_S+\kappa(h_S-r_\ell)$ hence using the definition of $\lambda_V$, equation \eqref{eq:def-lambdav}:
		\begin{align}
			\notag
			r_S-q-\kappa\lambda_V
			&=(h_S+\kappa(h_S-r_\ell))-q-\kappa(r_V-r_\ell)\\
			\label{eq2:proof-coro-pricing-pde-linearised}
			&=h_S-q+\kappa(h_S-r_V)
	\end{align}
	Moreover, using definitions \eqref{eq:recovery-value-1} and \eqref{eq:recovery-value-2} for the recovery processes, observe that:
	\begin{align}
		\notag
		\left(\sum_{i=1}^2\alpha_{i}z_i\right)^+
		&=\left((\alpha_1+\alpha_2\varkappa_2)(\widetilde{m})^+
		-(\alpha_1\varkappa_1+\alpha_2)(\widetilde{m})^-\right)^+\\
		&=(\alpha_1+\alpha_2\varkappa_2)(\widetilde{m})^+
	\end{align}
	given that $\varkappa_1,\varkappa_2>0$ and $\alpha_1,\alpha_2\geq0$ by assumption. Hence:
	\begin{align}
		\sum_{i=1}^2\lambda_iz_i-\varphi\left(\sum_{i=1}^2\alpha_{i}z_i\right)^+
		=(\lambda_1+\lambda_2\varkappa_2-\varphi(\alpha_1+\alpha_2\varkappa_2))(\widetilde{m})^+-(\lambda_1\varkappa_1+\lambda_2)(\widetilde{m})^-
	\end{align}
	The result ensues.
\end{proof}

\subsection{Proof of Proposition \ref{prop:valuation-forward}}
\label{app:proof-th-valuation-forward}

\begin{lemma}
	\label{lemm:bs-integrand}
	For any $u,t\in[0,T]$ such that $t<u$:
		\begin{align}
			\mathbb{E}_{t,s}^{\widehat{\mathbb{Q}}}\left(\left(F_T(u,(1+\kappa)S_{u-})-K\right)^+\right)
			&=\mathcal{BS}_C\left((1+\kappa)F_T(t,s)e^{\kappa(h_S-r_V)(u-t)},u-t\right)
		\end{align}
\end{lemma}

\begin{proof}
	Recall the definition of the risk-free forward price from equation \eqref{eq:rf-fwd-price-stock}:
		\begin{align}
			F_T(t,s):=se^{(h_S-q)(T-t)}
		\end{align}
	Moreover from Theorem \ref{th:conditional-expectation}, equation \eqref{eq:dynamics-s-pricing-measure}, the dynamics of $S$ under measure $\widehat{\mathbb{Q}}$ are:
		\begin{align}
			\label{eq1:proof-lemma-expectation-forward}
			\diff S_t=
			S_t((h_S-q+\kappa(h_S-r_V))\diff t+\sigma\diff \widehat{W}_t)
		\end{align}
	with initial condition $S_0$. Therefore under ${\widehat{\mathbb{Q}}}$ the price $S$ follows a geometric Brownian motion with drift $(h_S-q+\kappa(h_S-r_V))$ and diffusion term $\sigma$. Let $t\leq u\leq T$, we have $S_t=S_{t-}$ and:
		\begin{align}
			\notag
			F_T(u,S_u)
			&=S_ue^{(h_S-q)(T-u)}\\\notag
			&=\left(S_te^{(h_S-q)(T-u)}\right)e^{\left((h_S-q+\kappa(h_S-r_V))-\frac{\sigma^2}{2}\right)(u-t)
			+\sigma(\widehat{W}_u-\widehat{W}_t)}\\\notag
			&=\left(S_te^{(h_S-q)(T-t)}\right)e^{\left(\kappa(h_S-r_V)-\frac{\sigma^2}{2}\right)(u-t)
			+\sigma(\widehat{W}_u-\widehat{W}_t)}\\
			\label{eq2:proof-lemma-expectation-forward}
			&=F_T(t,S_t)e^{\left(\kappa(h_S-r_V)-\frac{\sigma^2}{2}\right)(u-t)
			+\sigma(\widehat{W}_u-\widehat{W}_t)}
		\end{align}
	Hence conditional on $\mathscr{F}_t$ the forward price $F_T(u,S_u)$ is log-normally distributed with mean and log-variance:
		\begin{align}
			\mathbb{E}_{t,s}^{\widehat{\mathbb{Q}}}(F_T(u,S_u))
			&=F_T(t,s)e^{\kappa(h_S-r_V)(u-t)}\\[3pt]
			V_{t,s}^{\widehat{\mathbb{Q}}}(\ln F_T(u,S_u))
			&=\sigma^2(u-t)
		\end{align}
	The jump factor $(1+\kappa)$ is a multiplicative constant to the expectation, hence we can apply Theorem \ref{th:black-scholes} and conclude.
\end{proof}

\begin{proof}[Proof of Proposition \ref{prop:valuation-forward}]
	We recall $f(s)=s-K$ and $m\equiv v^*$ where $v^*$ is defined in \eqref{eq:risk-free-forward}.
	\newline
	
	We start by computing the value from the terminal payoff component in the stochastic representation \eqref{eq:conditional-expectation}. Under measure $\widehat{\mathbb{Q}}$ the stock price has a log-normal distribution specified by the dynamics \eqref{eq:dynamics-s-pricing-measure} therefore:
	\begin{align}
		\label{eq1:proof-th-valuation-forward}
		\mathbb{E}_{t,s}^{\widehat{\mathbb{Q}}}
		\left(f(S_T)\right)
		=\mathbb{E}_{t,s}^{\widehat{\mathbb{Q}}}
		\left(S_T-K\right)
		=se^{(h_S-q+\kappa(h_S-r_V))(T-t)}-K
		=F_T(t,s)e^{\kappa(h_S-r_V)(T-t)}-K
	\end{align}
	where the last equality stems from the definition \eqref{eq:rf-fwd-price-stock} of $F$. Now we compute the integrands within  \eqref{eq:conditional-expectation}. Starting with the CVA integrand, we invert the integration order by linearity of expectation:
		\begin{align}
			\label{eq2:proof-th-valuation-forward}
			\mathbb{E}_{t,s}^{\widehat{\mathbb{Q}}}
			\left(
			\int_t^Te^{-r_V(u-t)}(\wthalf{M}_u)^+
			\diff u\right)
			=
			\int_t^Te^{-r_V(u-t)}
			\mathbb{E}_{t,s}^{\widehat{\mathbb{Q}}}
			((\wthalf{M}_u)^+)
			\diff u
		\end{align}
	Using the definition \eqref{eq:M-tilde} of $\wthalf{M}$ and our assumption that $m\equiv v^*$, notice the integrand in the right-hand side of \eqref{eq2:proof-th-valuation-forward} is equal to:
		\begin{align}
			\label{eq3:proof-th-valuation-forward}
			\mathbb{E}_{t,s}^{\widehat{\mathbb{Q}}}
			((\wthalf{M}_u)^+)
			=\mathbb{E}_{t,s}^{\widehat{\mathbb{Q}}}
			\left(\left(e^{-r(T-u)}(F_T(u,(1+\kappa)S_{u-})-K)\right)^+\right)
		\end{align}
	Using Lemma \ref{lemm:bs-integrand}, we immediately have that:
		\begin{align}
			\label{eq4:proof-th-valuation-forward}
			\mathbb{E}_{t,s}^{\widehat{\mathbb{Q}}}
			((\wthalf{M}_u)^+)
			=e^{-r(T-u)}\mathcal{BS}_C\left((1+\kappa)F_T(t,s)e^{\kappa(h_S-r_V)(u-t)},u-t\right)
		\end{align}
	Lemma \ref{lemm:bs-integrand} can be extended to the case of a put option therefore:
		\begin{align}
			\label{eq5:proof-th-valuation-forward}
			\mathbb{E}_{t,s}^{\widehat{\mathbb{Q}}}
			((\wthalf{M}_u)^-)
			=\mathbb{E}_{t,s}^{\widehat{\mathbb{Q}}}
			((-\wthalf{M}_u)^+)
			=\mathcal{BS}_P\left((1+\kappa)F_T(t,s)e^{\kappa(h_S-r_V)(u-t)},u-t\right)
		\end{align}
	Combining equations 
	\eqref{eq1:proof-th-valuation-forward},
	\eqref{eq4:proof-th-valuation-forward}, and 
	\eqref{eq5:proof-th-valuation-forward}
	yields the result.
\end{proof}

\subsection{Proof of Theorem \ref{th:analytical-fwd-value}}
\label{app:proof-th-analytical-fwd-value}
We start by proving two technical lemmas which give an analytical expression for certain classes of definite integrals involving the Gaussian integral \eqref{eq:gaussian-cdf} and its first derivative. We recall the expression for the \emph{Gaussian density} $\phi:\mathbb{R}\rightarrow\mathbb{R}$ of a standard normal random variable for any $x\in\mathbb{R}$:
	\begin{align}
		\label{eq:gaussian-density}
		\phi(x):=\frac{1}{\sqrt{2\pi}}e^{-\frac{1}{2}x^2}
	\end{align}

\begin{lemma}
	\label{lemm:A-1}
	Let $a,b\in\mathbb{R}$ and $t\geq0$. Define the real-valued integral $A_1:\mathbb{R}_{+}\times\mathbb{R}^2\rightarrow\mathbb{R}_{+}$:
	\begin{align}
		\label{eq:def-A-1}
		A_1(t,a,b):=
		\int_0^t\frac{1}{2\sqrt{u}}\phi\left(\left(au+\frac{b^2}{u}\right)^{\frac{1}{2}}\right)\diff u
	\end{align}
	Then, if $a>0$:
	\begin{align}
		\label{eq:sol-A-1-positive_a}
		A_1(t,a,b)=
		\begin{dcases}
			\frac{1}{2\sqrt{a}}\left(
			e^{b\sqrt{a}}\Phi\left(\frac{t\sqrt{a}+b}{\sqrt{t}}\right)
			-e^{-b\sqrt{a}}\Phi\left(\frac{-t\sqrt{a}+b}{\sqrt{t}}\right)
			\right), & b<0
			\\[4pt]
			\frac{1}{2\sqrt{a}}\left(
			e^{-b\sqrt{a}}\Phi\left(\frac{t\sqrt{a}-b}{\sqrt{t}}\right)
			-e^{b\sqrt{a}}\Phi\left(\frac{-t\sqrt{a}-b}{\sqrt{t}}\right)
			\right), & b\geq0
		\end{dcases}
	\end{align}
	On the other hand, if $a<0$:
	\begin{align}
		\label{eq:sol-A-1-negative_a}
		A_1(t,a,b)=
		\begin{dcases}
			\frac{1}{2i\sqrt{-a}}\left(
			e^{ib\sqrt{-a}}\Phi\left(\frac{it\sqrt{-a}+b}{\sqrt{t}}\right)
			-e^{-ib\sqrt{-a}}\Phi\left(\frac{-it\sqrt{-a}+b}{\sqrt{t}}\right)
			\right), & b<0
			\\[4pt]
			\frac{1}{2i\sqrt{-a}}\left(
			e^{-ib\sqrt{-a}}\Phi\left(\frac{it\sqrt{-a}-b}{\sqrt{t}}\right)
			-e^{ib\sqrt{-a}}\Phi\left(\frac{-it\sqrt{-a}-b}{\sqrt{t}}\right)
			\right), & b\geq0
		\end{dcases}
	\end{align}
	where $i$ is the imaginary unit. Finally, if $a=0$:
	\begin{align}
		\label{eq:sol-A-1-zero_a}
		A_1(t,0,b)=
		\begin{dcases}
			\sqrt{t}\phi\left(\frac{b}{\sqrt{t}}\right)
			+b\Phi\left(\frac{b}{\sqrt{t}}\right),
			&b<0
			\\[3pt]
			\sqrt{t}\phi\left(\frac{b}{\sqrt{t}}\right)
			-b\Phi\left(-\frac{b}{\sqrt{t}}\right),
			&b\geq0
		\end{dcases}
	\end{align}
\end{lemma}

\begin{proof}
	Throughout the proof we use the complementary symbols $\pm$ and $\mp$ to parametrize certain integrals and functions. For example, the notation $g_\pm(a,b,c):=(a\pm b\mp c)$ with $a,b,c\in\mathbb{R}$ is compact for two distinct functions, $g_+(a,b,c):=(a+b-c)$ and $g_-(a,b,c):=(a-b+c)$. When there is no ambiguity, we might shorten $g_\pm(a,b,c)$ to $g_\pm$.
	\newline
	
	We split our analysis based on whether $a$ is zero or not, then split the second case between positive and negative values for $a$.
	\newline
	
	\textbf{Case $(0,b)$}.
	The integral $A_1$ becomes, when $a=0$:
	\begin{align}
		\label{eq4:proof-lemma-A1}
		A_1(t,0,b)
		:=\frac{1}{\sqrt{2\pi}}\int_0^t\frac{1}{2\sqrt{u}}e^{-\frac{b^2}{2u}}\diff u
		=\int_0^t\frac{1}{2\sqrt{u}}\phi\left(\frac{b}{\sqrt{u}}\right)\diff u
	\end{align}
	Using integration by parts, we obtain the decomposition:
	\begin{align}
		\notag
		A_1
		&=\left[\sqrt{u}\phi\left(\frac{b}{\sqrt{u}}\right)\right]_{u=0}^t
		-\int_0^t\sqrt{u}\left(-\frac{b}{2u\sqrt{u}}\right)\left(-\frac{b}{\sqrt{u}}\phi\left(\frac{b}{\sqrt{u}}\right)\right)\diff u
		\\[4pt]\label{eq5:proof-lemma-A1}
		&=\sqrt{t}\phi\left(\frac{b}{\sqrt{t}}\right)
		+b\int_0^t\left(-\frac{b}{2u\sqrt{u}}\right)\phi\left(\frac{b}{\sqrt{u}}\right)\diff u
	\end{align}
	To solve the second integral in \eqref{eq5:proof-lemma-A1}, we make the change of variable $w=b/\sqrt{u}$. The analytical form of the resulting integral will depend on whether $b$ is positive or negative:
	\begin{align}
		\label{eq6:proof-lemma-A1}
		\int_0^t\left(-\frac{b}{2u\sqrt{u}}\right)\phi\left(\frac{b}{\sqrt{u}}\right)\diff u
		=\int_{\sgn(b)\times\infty}^{b/\sqrt{t}}\phi(w)\diff w
		=\begin{dcases}
			\Phi\left(\frac{b}{\sqrt{t}}\right),  	& b<0\\
			\Phi\left(\frac{b}{\sqrt{t}}\right)-1,  & b>0\\
		\end{dcases}
	\end{align}
	Collecting equations \eqref{eq5:proof-lemma-A1} and \eqref{eq6:proof-lemma-A1} together, we conclude for the case $(0,b)$:
	\begin{align}
		\label{eq7:proof-lemma-A1}
			A_1(t,0,b)=
			\begin{dcases}
				\sqrt{t}\phi\left(\frac{b}{\sqrt{t}}\right)
				+b\Phi\left(\frac{b}{\sqrt{t}}\right),
				&b<0
				\\[3pt]
				\sqrt{t}\phi\left(\frac{b}{\sqrt{t}}\right)
				-b\Phi\left(-\frac{b}{\sqrt{t}}\right),
				&b>0
			\end{dcases}
	\end{align}
	
	\textbf{Case $(a,b)$ with $a\in\mathbb{R}^*$.}
	We split our analysis in two parts, depending on whether $a$ is positive or negative.
	\newline
	
	\emph{Case $(a,b)$ with $a>0$}. We need to split $A_1$ into two distinct integrals we can then solve. To do so, observe the exponent within the integrand of $A_1$ admits two alternative factorizations: 
	\begin{align}
		\label{eq8:proof-lemma-A1}
		au+\frac{b^2}{u}
		=\left(\sqrt{au}\pm\frac{b}{\sqrt{u}}\right)^2\mp2b\sqrt{a}
	\end{align}
	Under the assumption that $a\neq0$, we decompose the integral $A_1$ as follows:
	\begin{align}
		\notag
		A_1&=
		\frac{2\sqrt{a}}{2\sqrt{a}}
		\int_0^t\frac{1}{2\sqrt{u}}\phi\left(\left(au+\frac{b^2}{u}\right)^{\frac{1}{2}}\right)\diff u
		\\[0pt]\notag
		&=\frac{1}{2\sqrt{a}}
		\int_0^t\frac{1}{2\sqrt{u}}\left(\left(\sqrt{a}+\frac{b}{u}\right)+\left(\sqrt{a}-\frac{b}{u}\right)\right)
		\phi\left(\left(au+\frac{b^2}{u}\right)^{\frac{1}{2}}\right)\diff u
		\\[10pt]\label{eq9:proof-lemma-A1}
		&=A_{1-}+A_{1+}
	\end{align}
	where we insert in \eqref{eq9:proof-lemma-A1} the real-valued integrals $A_{1+}$ and $A_{1-}:\mathbb{R}_+\times\mathbb{R}^2\rightarrow\mathbb{R}$ defined below:
	\begin{align}
		\label{eq10:proof-lemma-A1}
		A_{1\pm}(t,a,b)
		:=\frac{e^{\pm b\sqrt{a}}}{2\sqrt{a}}\int_0^t\frac{1}{2\sqrt{u}}\left(\sqrt{a}\mp\frac{b}{u}\right)
		\phi\left(\sqrt{au}\pm\frac{b}{\sqrt{u}}\right)\diff u
	\end{align}
	To solve the integrals $A_{1\pm}$ it suffices to make the change of variable $w_\pm=\sqrt{au}\pm\frac{b}{\sqrt{u}}$ that is:
	\begin{align}
		\label{eq11:proof-lemma-A1}
		\diff w_\pm 
		=\frac{1}{2\sqrt{u}}\left(\sqrt{a}\mp\frac{b}{u}\right)\diff u
	\end{align}
	We obtain an expression for both in terms of elementary functions and the Gaussian CDF $\Phi$:
	\begin{align}
		\label{eq12:proof-lemma-A1}
		A_{1+}
		=\frac{e^{b\sqrt{a}}}{2\sqrt{a}}\int_{\sgn(b)\times\infty}^{\sqrt{at}+b/\sqrt{t}}\phi(w)\diff w
		=\begin{dcases}
			\frac{e^{b\sqrt{a}}}{2\sqrt{a}}\Phi\left(\sqrt{at}+\frac{b}{\sqrt{t}}\right), 				 &b<0\\[4pt]
			\frac{e^{b\sqrt{a}}}{2\sqrt{a}}\left(\Phi\left(\sqrt{at}+\frac{b}{\sqrt{t}}\right)-1\right), &b>0
		\end{dcases}
		\\[8pt]
		\label{eq13:proof-lemma-A1}
		A_{1-}
		=\frac{e^{-b\sqrt{a}}}{2\sqrt{a}}\int_{-\sgn(b)\times\infty}^{\sqrt{at}-b/\sqrt{t}}\phi(w)\diff w
		=\begin{dcases}
			\frac{e^{-b\sqrt{a}}}{2\sqrt{a}}\left(\Phi\left(\sqrt{at}-\frac{b}{\sqrt{t}}\right)-1\right), &b<0\\[4pt]
			\frac{e^{-b\sqrt{a}}}{2\sqrt{a}}\Phi\left(\sqrt{at}-\frac{b}{\sqrt{t}}\right), 				  &b>0
		\end{dcases}
	\end{align}
	Inserting \eqref{eq12:proof-lemma-A1} and \eqref{eq13:proof-lemma-A1} into \eqref{eq9:proof-lemma-A1} we obtain the expression for $A_1$ when $a>0$:
	\begin{align}
		\label{eq14:proof-lemma-A1}
			A_1(t,a,b)=
			\begin{dcases}
				\frac{1}{2\sqrt{a}}\left(
				e^{b\sqrt{a}}\Phi\left(\frac{t\sqrt{a}+b}{\sqrt{t}}\right)
				-e^{-b\sqrt{a}}\Phi\left(\frac{-t\sqrt{a}+b}{\sqrt{t}}\right)
				\right), & b<0
				\\[4pt]
				\frac{1}{2\sqrt{a}}\left(
				e^{-b\sqrt{a}}\Phi\left(\frac{t\sqrt{a}-b}{\sqrt{t}}\right)
				-e^{b\sqrt{a}}\Phi\left(\frac{-t\sqrt{a}-b}{\sqrt{t}}\right)
				\right), & b>0
			\end{dcases}
	\end{align}
	
	\emph{Case $(a,b)$ with $a<0$}. On the other hand, we need to replace the factorization in \eqref{eq8:proof-lemma-A1} by one that uses imaginary numbers: 
	\begin{align}
		\label{eq15:proof-lemma-A1}
		au+\frac{b^2}{u}
		=\left(i\sqrt{-au}\pm\frac{b}{\sqrt{u}}\right)^2\mp2ib\sqrt{-a}
	\end{align}
	When $a>0$ we define the integrals $A_{1\pm}$ to be equal to:
	\begin{align}
		\label{eq16:proof-lemma-A1}
		A_{1\pm}
		:=\frac{e^{\pm ib\sqrt{-a}}}{2\sqrt{-a}}\int_0^t\frac{1}{2\sqrt{u}}\left(i\sqrt{-a}\mp\frac{b}{u}\right)
		\phi\left(i\sqrt{-au}\pm\frac{b}{\sqrt{u}}\right)\diff u
	\end{align}
	Proceeding in the same way as for $a>0$ but replacing $\sqrt{a}$ by $i\sqrt{-a}$, we show $A_1=A_{1+}+A_{1-}$ and make the change of variable $w_\pm=i\sqrt{-au}\pm\frac{b}{\sqrt{u}}$, thereby obtaining the following values:
	\begin{align}
		\label{eq17:proof-lemma-A1}
		A_{1+}
		=\frac{e^{ib\sqrt{-a}}}{2i\sqrt{-a}}\int_{\sgn(b)\times\infty}^{i\sqrt{-at}+b/\sqrt{t}}\phi(w)\diff w
		=\begin{dcases}
			\frac{e^{ib\sqrt{-a}}}{2i\sqrt{-a}}\Phi\left(i\sqrt{-at}+\frac{b}{\sqrt{t}}\right), 				&b<0\\[4pt]
			\frac{e^{ib\sqrt{-a}}}{2i\sqrt{-a}}\left(\Phi\left(i\sqrt{-at}+\frac{b}{\sqrt{t}}\right)-1\right),  &b>0
		\end{dcases}
		\\[8pt]
		\label{eq18:proof-lemma-A1}
		A_{1-}
		=\frac{e^{-ib\sqrt{-a}}}{2i\sqrt{-a}}\int_{-\sgn(b)\times\infty}^{i\sqrt{-at}-b/\sqrt{t}}\phi(w)\diff w
		=\begin{dcases}
			\frac{e^{-ib\sqrt{-a}}}{2i\sqrt{-a}}\left(\Phi\left(i\sqrt{-at}-\frac{b}{\sqrt{t}}\right)-1\right), &b<0\\[4pt]
			\frac{e^{-ib\sqrt{-a}}}{2i\sqrt{-a}}\Phi\left(i\sqrt{-at}-\frac{b}{\sqrt{t}}\right), 				&b>0
		\end{dcases}
	\end{align}	
	which allows us to conclude when $a>0$ -- and the proof:
	\begin{align}
		\label{eq19:proof-lemma-A1}
			A_1(t,a,b)=
			\begin{dcases}
				\frac{1}{2i\sqrt{-a}}\left(
				e^{ib\sqrt{-a}}\Phi\left(\frac{it\sqrt{-a}+b}{\sqrt{t}}\right)
				-e^{-ib\sqrt{-a}}\Phi\left(\frac{-it\sqrt{-a}+b}{\sqrt{t}}\right)
				\right), & b<0
				\\[4pt]
				\frac{1}{2i\sqrt{-a}}\left(
				e^{-ib\sqrt{-a}}\Phi\left(\frac{it\sqrt{-a}-b}{\sqrt{t}}\right)
				-e^{ib\sqrt{-a}}\Phi\left(\frac{-it\sqrt{-a}-b}{\sqrt{t}}\right)
				\right), & b>0
			\end{dcases}
	\end{align}
\end{proof}

\begin{lemma}
	\label{lemm:A-2}
	Let $a\in\mathbb{R}$, $b\in\mathbb{R}^*$ and $t\geq0$. Define the real-valued integral $A_2:\mathbb{R}_{+}\times\mathbb{R}\times\mathbb{R}^*\rightarrow\mathbb{R}_{+}$:
	\begin{align}
		\label{eq:def-A-2}
		A_2(t,a,b):=
		\int_0^t\frac{1}{2u\sqrt{u}}
		\phi\left(\left(au+\frac{b^2}{u}\right)^{\frac{1}{2}}\right)
		\diff u
	\end{align}
	Then, if $a>0$:
	\begin{align}
		\label{eq:sol-A-2-positive_a}
		A_2(t,a,b)=
		\begin{dcases}
			\frac{1}{2b}\left(
			e^{b\sqrt{a}}\Phi\left(\frac{t\sqrt{a}+b}{\sqrt{t}}\right)
			-e^{-b\sqrt{a}}\Phi\left(\frac{-t\sqrt{a}+b}{\sqrt{t}}\right)
			\right), 
			& b<0
			\\[4pt]
			\frac{1}{2b}\left(
			e^{-b\sqrt{a}}\Phi\left(\frac{t\sqrt{a}-b}{\sqrt{t}}\right)
			+e^{b\sqrt{a}}\Phi\left(\frac{-t\sqrt{a}-b}{\sqrt{t}}\right)
			\right), 
			& b>0 
		\end{dcases}
	\end{align}
	On the other hand, if $a<0$:
	\begin{align}
		\label{eq:sol-A-2-negative_a}
		A_2(t,a,b)=
		\begin{dcases}
			\frac{1}{2b}\left(
			e^{ib\sqrt{-a}}\Phi\left(\frac{it\sqrt{-a}+b}{\sqrt{t}}\right)
			-e^{-ib\sqrt{-a}}\Phi\left(\frac{-it\sqrt{-a}+b}{\sqrt{t}}\right)
			\right), 
			& b<0 
			\\[4pt]
			\frac{1}{2b}\left(
			e^{-ib\sqrt{-a}}\Phi\left(\frac{it\sqrt{-a}-b}{\sqrt{t}}\right)
			+e^{ib\sqrt{-a}}\Phi\left(\frac{-it\sqrt{-a}-b}{\sqrt{t}}\right)
			\right), 
			& b>0
		\end{dcases}
	\end{align}
	Finally, if $a=0$:
	\begin{align}
		\label{eq:sol-A-2-zero_a}
		A_2(t,0,b)=
		\begin{dcases}
			-\frac{1}{b}\Phi\left(\frac{b}{\sqrt{t}}\right),  & b<0\\[4pt]
			\frac{1}{b}\Phi\left(-\frac{b}{\sqrt{t}}\right), & b>0 
		\end{dcases}
	\end{align}
\end{lemma}

\begin{proof}
	We use the same conventions as in Lemma \ref{lemm:A-1}. We need to restrict ourselves to the case where $b$ is non-zero as otherwise the integral does not converge. Starting with $a>0$:
	\begin{align}
		\notag
		A_2(t,a,b)
		&=\frac{1}{\sqrt{2\pi}}\left(\frac{2b}{2b}\right)\int_0^t\frac{1}{2u\sqrt{u}}
		e^{-\frac{1}{2}\left(au+\frac{b^2}{u}\right)}\diff u
		\\\notag
		&=\frac{1}{\sqrt{2\pi}}\left(\frac{1}{2b}\right)\int_0^t\frac{1}{2\sqrt{u}}
		\left(\left(\sqrt{a}+\frac{b}{u}\right)-\left(\sqrt{a}-\frac{b}{u}\right)\right)
		e^{-\frac{1}{2}\left(au+\frac{b^2}{u}\right)}\diff u
		\\[5pt]\label{eq1:proof-lemma-A2}
		&=\frac{\sqrt{a}}{b}\left(A_{1-}(t,a,b)-A_{1+}(t,a,b)\right)
	\end{align}
	which allows us to conclude using the results from Lemma \ref{lemm:A-1}. For the case $a<0$, the expression for $A_2$ is simply:
	\begin{align}
		\label{eq2:proof-lemma-A2}
		A_2(t,a,b)=
		\frac{i\sqrt{-a}}{b}\left(A_{1-}(t,a,b)-A_{1+}(t,a,b)\right)
	\end{align}
	For the case $a=0$, it suffices to write instead:
	\begin{align}
		\notag
		A_2(t,0,b)
		&=\frac{1}{b\sqrt{2\pi}}\int_0^t\frac{b}{2u\sqrt{u}}
		e^{-\frac{1}{2}\left(\frac{b^2}{u}\right)}\diff u
	\end{align}
	then use equation \eqref{eq6:proof-lemma-A1} from the derivation of the case $(0,b)$ in the proof of Lemma \ref{lemm:A-1} in order to conclude.
\end{proof}

\begin{proposition}
	\label{prop:Lambda}
	Let $x,y,z\in\mathbb{R}$ and $t\geq0$ such that $x\neq0$. Define the real-valued integral $\Lambda:\mathbb{R}_{+}\times\mathbb{R}^*\times\mathbb{R}^2\rightarrow\mathbb{R}_{+}$:
	\begin{align}
		\label{eq:def-Lambda}
		\Lambda(t,x,y,z):=
		\int_0^te^{-xu}\Phi\left(y\sqrt{u}+\frac{z}{\sqrt{u}}\right)\diff u
	\end{align}
	Define the complex-valued function $\rho:\mathbb{R}^*\times\mathbb{R}\rightarrow\mathbb{C}$:
	\begin{align}
		\rho(x,y):=
			\begin{dcases}
				\sqrt{2x+y^2}, 		& 2x+y^2\geq0\\
				i\sqrt{|2x+y^2|},	& 2x+y^2<0
			\end{dcases}
	\end{align}
	as well as the variables:
	\begin{align}
		\label{eq:def-betas}
		\beta_0:=\frac{yt+z}{\sqrt{t}},\qquad
		\beta_1:=\frac{\rho(x,y)t-z}{\sqrt{t}},\qquad
		\beta_2:=\beta_1+\frac{2z}{\sqrt{t}}
	\end{align}
	Letting $\rho:=\rho(x,y)$, if $\rho(x,y)\neq0$ then $\Lambda$ admits the analytical expression:
	\begin{align}
		\notag
		&\Lambda(t,x,y,z)=\\[4pt]
		\label{eq:Lambda-solution-rho-nonzero}
		&\
		\begin{dcases}
			\frac{1}{x}\biggl(
			-e^{-xt}\Phi\left(\beta_0\right)
			-\frac{e^{-yz}}{2}\left(
			e^{-z\rho}\left(\frac{y}{\rho}-1\right)
			\Phi\left(-\beta_1\right)
			-e^{z\rho}\left(\frac{y}{\rho}+1\right)
			\Phi\left(\beta_2\right)
			\right)\biggr),
			& z<0
			\\[4pt]
			\frac{1}{x}\biggl(
			1-e^{-xt}\Phi\left(\beta_0\right)
			+\frac{e^{-yz}}{2}\left(
			e^{-z\rho}\left(\frac{y}{\rho}-1\right)
			\Phi\left(\beta_1\right)
			-e^{z\rho}\left(\frac{y}{\rho}+1\right)
			\Phi\left(-\beta_2\right)
			\right)\biggr), 
			& z\geq0 
		\end{dcases}
	\end{align}
\end{proposition}

\begin{proof}
	We will use integration by parts to decompose the parametrized integral $\Lambda$ into different components, then apply Lemmas \ref{lemm:A-1} and \ref{lemm:A-2} to prove the analytical expression \eqref{eq:Lambda-solution-rho-nonzero} -- we distinguish between the cases $z\in\mathbb{R}^*$ and $z=0$.
	\newline
	
	Integrating by parts $\Lambda$, we obtain the following expression:
	\begin{align}
		\label{eq1:proof-prop-Lambda}
		\Lambda=
		\left[-\frac{e^{-xu}}{x}\Phi\left(y\sqrt{u}+\frac{z}{\sqrt{u}}\right)
		\right]_{u=0}^t
		+\frac{1}{x}\int_0^te^{-xu}\left(
		\frac{1}{2\sqrt{u}}\left(y-\frac{z}{u}\right)
		\phi\left(y\sqrt{u}+\frac{z}{\sqrt{u}}\right)
		\right)\diff u
	\end{align}
	Observe that the first term in equation \eqref{eq1:proof-prop-Lambda} depends on the sign of $z$:
	\begin{align}
		\label{eq2:proof-prop-Lambda}
		\lim_{u\searrow0}\Phi\left(y\sqrt{u}+\frac{z}{\sqrt{u}}\right)=
		\begin{dcases}
			0,			\quad 	& z<0\\[2pt]
			\frac{1}{2},\quad 	& z=0\\[2pt]
			1,			\quad 	& z>0
		\end{dcases}
	\end{align}
	Hence the first term in \eqref{eq1:proof-prop-Lambda} is equal to:
	\begin{align}
		\label{eq3:proof-prop-Lambda}
		\left[-\frac{e^{-xu}}{x}\Phi\left(y\sqrt{u}+\frac{z}{\sqrt{u}}\right)
		\right]_{u=0}^t=
		\begin{dcases}
			-\frac{e^{-xt}}{x}
			\Phi\left(y\sqrt{t}+\frac{z}{\sqrt{t}}\right),
			\quad & z<0
			\\[5pt]
			\frac{1}{x}\left(\frac{1}{2}
			-e^{-xt}\Phi(y\sqrt{t})\right),
			\quad & z=0
			\\[5pt]
			\frac{1}{x}\left(1-e^{-xt}
			\Phi\left(y\sqrt{t}+\frac{z}{\sqrt{t}}\right)
			\right),
			\quad & z>0
		\end{dcases}
	\end{align}
	
	\textbf{Case $z\in\mathbb{R}^*$.}
	If $z\neq0$, we use the integrals $A_1$ and $A_2$ introduced in Lemmas \ref{lemm:A-1} and \ref{lemm:A-2} respectively to rewrite the second term in \eqref{eq1:proof-prop-Lambda}. To do so, first observe that the exponential terms within the integrand are equal to:
	\begin{align}
		\label{eq4:proof-prop-Lambda}
		e^{-xu}
		\phi\left(y\sqrt{u}+\frac{z}{\sqrt{u}}\right)
		=e^{-yz}
		\phi\left(\left((2x+y^2)u+\frac{z^2}{u}\right)^{\frac{1}{2}}\right)
	\end{align}
	Therefore the second integral in \eqref{eq1:proof-prop-Lambda} is equal to:
	\begin{multline}
		\label{eq5:proof-prop-Lambda}
		\int_0^te^{-xu}\left(
		\frac{1}{2\sqrt{u}}\left(y-\frac{z}{u}\right)
		\phi\left(y\sqrt{u}+\frac{z}{\sqrt{u}}\right)
		\right)\diff u\\
		=e^{-yz}\left(yA_1(t,2x+y^2,z)-zA_2(t,2x+y^2,z)\right)
	\end{multline}
	Collecting equations 
	\eqref{eq3:proof-prop-Lambda} and
	\eqref{eq5:proof-prop-Lambda},
	we write \eqref{eq2:proof-prop-Lambda} as follows for the case $z\in\mathbb{R}^*$:
	\begin{multline}
		\label{eq6:proof-prop-Lambda}
		\Lambda(t,x,y,z)=\\
		\frac{1}{x}\left(
		\indicatorf{(0,\infty)}{z}
		-e^{-xt}\Phi\left(\frac{yt+z}{\sqrt{t}}\right)
		+e^{-yz}\left(yA_1(t,2x+y^2,z)-zA_2(t,2x+y^2,z)\right)
		\right)
	\end{multline} 
	Replacing $A_1$ and $A_2$ by their values from Lemmas \ref{lemm:A-1} and \ref{lemm:A-2} gives the result. 
	\newline
	
	\textbf{Case $z=0$.}
	On the other hand, when $z$ is zero then the second integral in equation \eqref{eq1:proof-prop-Lambda} reduces to:
	\begin{align}
		\label{eq7:proof-lemm-Lambda}
		y\int_0^te^{-xu}
		\frac{1}{2\sqrt{u}}\phi\left(y\sqrt{u}\right)
		\diff u
		=y\int_0^t
		\frac{1}{2\sqrt{u}}
		\phi\left(\sqrt{(2x+y^2)u}\right)
		\diff u
		=yA_1(t,2x+y^2,0)
	\end{align}
	Therefore if $z=0$ the integral $\Lambda$ is equal to:
	\begin{align}
		\label{eq8:proof-lemm-Lambda}
		\Lambda(t,x,y,0)=\frac{1}{x}\left(
		\frac{1}{2}-e^{-xt}\Phi(y\sqrt{t})
		+yA_1(t,2x+y^2,0)\right)
	\end{align}
	One can verify the function $\Lambda$ is continuous at $z=0$:
	\begin{align}
		\label{eq9:proof-lemm-Lambda}
		\lim_{z\nearrow0}\Lambda(t,x,y,z)
		=\lim_{z\searrow0}\Lambda(t,x,y,z)
		=\Lambda(t,x,y,0)
	\end{align}
	which concludes the proof.
\end{proof}

\begin{corollary}
	\label{coro:Lambda-zero-rho}
	Let $x,y,z\in\mathbb{R}$ and $t\geq0$ such that $x\neq0$ and $\rho(x,y)=0$. Then $\Lambda$ admits the analytical expression:
	\begin{align}
		\notag
		&\Lambda(t,x,y,z)=\\[4pt]
		\label{eq:Lambda-solution-rho-zero}
		&\
		\begin{dcases}
			\frac{1}{x}\biggl(
			-e^{-xt}\Phi\left(\beta_0\right)
			+e^{-yz}\left(
			y\sqrt{t}\phi\left(\frac{z}{\sqrt{t}}\right)
			+(1+yz)\Phi\left(\frac{z}{\sqrt{t}}\right)
			\right)\biggr),
			& z<0
			\\[4pt]
			\frac{1}{x}\biggl(
			1-e^{-xt}\Phi\left(\beta_0\right)
			+e^{-yz}\left(
			y\sqrt{t}\phi\left(\frac{z}{\sqrt{t}}\right)
			-(1+yz)\Phi\left(-\frac{z}{\sqrt{t}}\right)
			\right)\biggr), 
			& z\geq0 
		\end{dcases}
	\end{align}
\end{corollary}

\begin{proof}
	Following the same steps as in the proof of Proposition \ref{prop:Lambda}, function $\Lambda$ satisfies equation \eqref{eq6:proof-prop-Lambda} even if $\rho(x,y)=0$. It then suffices to replace $A_1(t,0,z)$ and $A_2(t,0,z)$ by the values obtained using equations \eqref{eq:sol-A-1-zero_a} from Lemma \ref{lemm:A-1} and \eqref{eq:sol-A-2-zero_a} from Lemma \ref{lemm:A-2} in order to conclude.
\end{proof}

\begin{corollary}
	\label{coro:Lambda-zero-x}
	Let $x=0$, $y,z\in\mathbb{R}$ and $t\geq0$ such that $y\neq0$. Then $\Lambda$ admits the analytical expression:
	\begin{align}
		\label{eq:Lambda-solution-x-zero}
		\Lambda(t,0,y,z)=
		\begin{dcases}
			t\Phi(\beta_0)+\frac{1}{2y}\left(
			2\sqrt{t}\phi(\beta_0)
			+\frac{e^{-2zy}}{y}\Phi(-\beta_1)
			+\left(2z-\frac{1}{y}\right)\Phi(\beta_0)
			\right),
			& z<0
			\\[4pt]
			t\Phi(\beta_0)+\frac{1}{2y}\left(
			2\sqrt{t}\phi(\beta_0)
			-\frac{e^{-2zy}}{y}\Phi(\beta_1)
			-\left(2z-\frac{1}{y}\right)\Phi(-\beta_0)
			\right), 
			& z\geq0 
		\end{dcases}
	\end{align}
\end{corollary}

\begin{proof}
	We use L'H\^opital rule applied to the expression \eqref{eq:Lambda-solution-rho-nonzero} from Proposition \ref{prop:Lambda}. We show detailed steps for the case $z\geq0$; the proof for the case where $z<0$ proceeds analogously.
	\newline
	
	\textbf{Case $z\geq0$.}
	Define the function $g_1:\mathbb{R}\rightarrow\mathbb{R}$ as follows for $y\neq0$:
		\begin{align}
			g_1(x):=
			1-e^{-xt}\Phi\left(\beta_0\right)
			+\frac{e^{-yz}}{2}\left(
			e^{-z\rho}\left(\frac{y}{\rho}-1\right)
			\Phi\left(\beta_1\right)
			-e^{z\rho}\left(\frac{y}{\rho}+1\right)
			\Phi\left(-\beta_2\right)
			\right)
		\end{align}
	Then $\Lambda(t,x,y,z)=g_1(x)/x$ when $z\geq0$. We have $\rho(0,y)=y$ hence:
		\begin{align}
			g_1(0)
			=1-\Phi\left(\beta_0\right)-\Phi\left(-\beta_2\right)
			=0
		\end{align}
	where the second equality stems from $\beta_2=\beta_0$ when $x=0$. Hence the limit of $\Lambda$ at $x=0$ is an indeterminate form of type $\frac{0}{0}$. Given $g_1$ is differentiable and that:
		\begin{align}
			\frac{\partial\Lambda}{\partial x}(t,x,y,z)
			=g_1^\prime(x),
		\end{align}
	the conditions for the application of L'H\^opital rule are fulfilled. By standard algebraic operations and using the notation $\rho=\rho(x,y)$, we obtain:
		\begin{multline}
			\label{eq1:proof-Lambda-solution-x-zero}
			g_1^\prime(x)=
			te^{-xt}\Phi(\beta_0)+\frac{e^{-yz}}{2}\Bigg\{
				e^{-z\rho}\left(
					\frac{1}{\rho}\left(\frac{y}{\rho}-1\right)\left(
						\sqrt{t}\phi(\beta_1)-z\Phi(\beta_1)
					\right)
					-\frac{y}{\rho^3}\Phi(\beta_1)
				\right)\\
				-e^{z\rho}\left(
				\frac{1}{\rho}\left(\frac{y}{\rho}+1\right)\left(
				z\Phi(-\beta_2)-\sqrt{t}\phi(-\beta_2)
				\right)
				-\frac{y}{\rho^3}\Phi(-\beta_2)
				\right)
			\Bigg\}
		\end{multline}
	where we have used the following equalities (interpreting $\beta_1,\beta_2$ as functions of $x$):
		\begin{align}
			\frac{\partial\rho}{\partial x}(x,y)
			=\frac{1}{\rho(x,y)},\qquad
			\beta_1^\prime(x)
			=\beta_2^\prime(x)
			=\frac{\sqrt{t}}{\rho(x,y)}
		\end{align}
	Evaluating expression \eqref{eq1:proof-Lambda-solution-x-zero} at $x=0$ we obtain:
		\begin{align}
			g_1^\prime(0)=t\Phi(\beta_0)+\frac{1}{2y}\left(
			2\sqrt{t}\phi(-\beta_2)
			-\frac{e^{-2zy}}{y}\Phi(\beta_1)
			-\left(2z-\frac{1}{y}\right)\Phi(-\beta_2)
			\right)
		\end{align}
	Using the identity $\beta_2=\beta_0$ when $x=0$ and that the function $\phi$ is symmetric around zero, we conclude using L'H\^opital rule for the case $z\geq0$:
		\begin{align}
			\notag
			\Lambda(t,0,y,z)
			&=\lim_{x\rightarrow0}g_1^\prime(x)
			\\[4pt]
			&=t\Phi(\beta_0)+\frac{1}{2y}\left(
			2\sqrt{t}\phi(\beta_0)
			-\frac{e^{-2zy}}{y}\Phi(\beta_1)
			-\left(2z-\frac{1}{y}\right)\Phi(-\beta_0)
			\right)
		\end{align}
	
	\textbf{Case $z<0$.}
	By applying the same steps as for the case $z\geq0$ but using function $g_2:\mathbb{R}\rightarrow\mathbb{R}$ defined as follows for $y\neq0$:
	\begin{align}
		g_2(x):=
		-e^{-xt}\Phi\left(\beta_0\right)
		-\frac{e^{-yz}}{2}\left(
		e^{-z\rho}\left(\frac{y}{\rho}-1\right)
		\Phi\left(-\beta_1\right)
		-e^{z\rho}\left(\frac{y}{\rho}+1\right)
		\Phi\left(\beta_2\right)
		\right),
	\end{align}
	we obtain the following expression for $\Lambda$ whenever $z<0$ (and $y\neq0$):
	\begin{align}
		\notag
		\Lambda(t,0,y,z)
		&=\lim_{x\rightarrow0}g_2^\prime(x)
		\\[4pt]
		&=t\Phi(\beta_0)+\frac{1}{2y}\left(
		2\sqrt{t}\phi(\beta_0)
		+\frac{e^{-2zy}}{y}\Phi(-\beta_1)
		+\left(2z-\frac{1}{y}\right)\Phi(\beta_0)
		\right)
	\end{align}
	which concludes the proof.
\end{proof}

\begin{corollary}
	\label{coro:Lambda-zero-xy}
	Let $x=0$, $y=0$, $z\in\mathbb{R}$ and $t\geq0$. Then $\Lambda$ admits the analytical expression:
	\begin{align}
		\label{eq:Lambda-solution-xy-zero}
		\Lambda(t,0,0,z)=
		\begin{dcases}
			(t+z^2)\Phi\left(\frac{z}{\sqrt{t}}\right)
			+z\sqrt{t}\phi\left(\frac{z}{\sqrt{t}}\right),
			& z<0
			\\[4pt]
			(t+z^2)\Phi\left(\frac{z}{\sqrt{t}}\right)
			+z\sqrt{t}\phi\left(\frac{z}{\sqrt{t}}\right)-z^2, 
			& z\geq0 
		\end{dcases}
	\end{align}
\end{corollary}

\begin{proof}
	To prove equation \eqref{eq:Lambda-solution-xy-zero} it suffices to make the following integration by parts:
		\begin{align}
			\int_0^t\Phi\left(\frac{z}{\sqrt{u}}\right)\diff u
			=\left[u\Phi\left(\frac{z}{\sqrt{u}}\right)\right]_{u=0}^t
			+z\int_0^t\frac{1}{2\sqrt{u}}\phi\left(\frac{z}{\sqrt{u}}\right)\diff u,
		\end{align}
	where we recognize the second integral to be $A_1(t,0,z)$, and use the results from Lemma \ref{lemm:A-1} for the case $a=0$, equation \eqref{eq:sol-A-1-zero_a}, to conclude.
\end{proof}

\begin{proof}[Proof of Theorem \ref{th:analytical-fwd-value}]
	To prove the analytical pricing formula \eqref{eq:analytical-fwd-value} for the vulnerable forward contract, we only need to show both recovery integrals within the option portfolio representation \eqref{eq:valuation-forward} can be expressed using the integral defined in \eqref{eq:def-Lambda}. To compact notation, we define the \emph{adjusted forward price} $\widetilde{F}_T:[0,T]\times\mathbb{R}_+^*\rightarrow\mathbb{R}$:
	\begin{align}
		\label{eq:def-vulnerable-fwd-price}
		\widetilde{F}_T(t,s)
		:=(1+\kappa)F_T(t,s)
	\end{align}
	
	We designate by $I_C$ the CVA integral involved in equation \eqref{eq:valuation-forward}, and we make the change of variables $\omega=u-t$:
		\begin{multline}
			\label{eq1:proof-fwd-analytical-price}
			I_C:=\int_t^T
			e^{-r_V(u-t)}e^{-r(T-u)}
			\mathcal{BS}_C\left(\widetilde{F}_T(t,s)
			e^{\kappa(h_S-r_V)(u-t)},u-t
			\right)\diff u\\
			=e^{-r(T-t)}\int_0^{T-t}
			e^{-(r_V-r)\omega}
			\mathcal{BS}_C\left(\widetilde{F}_T(t,s)
			e^{\kappa(h_S-r_V)\omega},\omega
			\right)\diff\omega
		\end{multline}
	Letting $\widetilde{F}_T:=\widetilde{F}_T(t,s)$ and using the definition of the Black-Scholes call function from \eqref{eq:bs-call-formula}, the Black-Scholes integrand in \eqref{eq1:proof-fwd-analytical-price} is equal to:
		\begin{multline}
			\label{eq3:proof-fwd-analytical-price}
			\mathcal{BS}_C\left(
			\widetilde{F}_Te^{\kappa(h_S-r_V)\omega},
			\omega\right)\\
			=\widetilde{F}_Te^{\kappa(h_S-r_V)\omega}
			\Phi\left(d_1(
			\widetilde{F}_Te^{\kappa(h_S-r_V)\omega},
			\omega)\right)
			-K\Phi\left(d_2(
			\widetilde{F}_Te^{\kappa(h_S-r_V)\omega},
			\omega)\right)
		\end{multline}
	Using the expression for function $d_1$ from equation \eqref{eq:def-d1} and the definition \eqref{eq:moneyness} of the moneyness $\eta(s)$, the value of $d_1$ in \eqref{eq3:proof-fwd-analytical-price} is equal to:
		\begin{align}
			\notag
			d_1(\widetilde{F}_T
			e^{\kappa(h_S-r_V)\omega},\omega)
			&=\frac{1}{\sigma\sqrt{\omega}}\left(
			\ln\left(\frac{\widetilde{F}_Te^{\kappa(h_S-r_V)\omega}}{K}\right)+\frac{1}{2}\sigma^2\omega
			\right)
			\\[4pt]\notag
			&=\frac{1}{\sigma\sqrt{\omega}}\left(
			\sigma\eta(s)+\left(\kappa(h_S-r_V)+\frac{1}{2}\sigma^2\right)\omega
			\right)
			\\[4pt]\label{eq4:proof-fwd-analytical-price}
			&=\frac{\eta(s)}{\sqrt{\omega}}+
			\left(\frac{2\kappa(h_S-r_V)+\sigma^2}{2\sigma}\right)\sqrt{\omega}
		\end{align}
	Therefore, using the definition \eqref{eq:def-d2} for $d_2$, the integral \eqref{eq1:proof-fwd-analytical-price} is equal to:
		\begin{multline}
			\label{eq5:proof-fwd-analytical-price}
			I_C=e^{-r(T-t)}\Biggl(
			\widetilde{F}_T\int_0^{T-t}
			e^{-(r_V-r-\kappa(h_S-r_V))\omega}
			\Phi\left(\frac{\eta(s)}{\sqrt{\omega}}
			+\left(
			\frac{2\kappa(h_S-r_V)+\sigma^2}{2\sigma}
			\right)\sqrt{\omega}\right)\diff\omega
			\\
			-
			K\int_0^{T-t}e^{-(r_V-r)\omega}
			\Phi\left(\frac{\eta(s)}{\sqrt{\omega}}
			+\left(
			\frac{2\kappa(h_S-r_V)-\sigma^2}{2\sigma}
			\right)\sqrt{\omega}\right)\diff\omega
			\Biggr)
		\end{multline}
	To lighten notation, we define the variables $\zeta_1$ and $\zeta_2$ such that:
		\begin{align}
			\label{eq7:proof-fwd-analytical-price}
			\zeta_1:=
			\frac{2\kappa(h_S-r_V)+\sigma^2}
			{2\sigma},\qquad
			\zeta_2:=
			\zeta_1-\sigma
		\end{align}
	If the following conditions are met:
		\begin{align}
			\label{eq10:proof-fwd-analytical-price}
			r_V-r-\kappa(h_S-r_V)\neq0,\qquad
			r_V-r\neq0
		\end{align}
	Then the assumptions for Proposition \ref{prop:Lambda} hold with the substitutions:
		\begin{align}
			x = r_V-r-\kappa(h_S-r_V), \qquad
			y = \zeta_1, \qquad
			z = \eta(s)
		\end{align}
	for the first integral in \eqref{eq5:proof-fwd-analytical-price}, and:
		\begin{align}
			x = r_V-r, \qquad
			y = \zeta_2, \qquad
			z = \eta(s)
		\end{align}
	for the second one. Integral $I_C$ is equal to:
		\begin{align}
			\label{eq11:proof-fwd-analytical-price}
			I_C=e^{-r(T-t)}\Bigl(
			\widetilde{F}_T
			\Lambda\left(T-t,(\lambda^*_V-\kappa(h_S-r_V)),
			\zeta_1,
			\eta\right)
			-K\Lambda\left(T-t,\lambda^*_V,
			\zeta_2,\eta\right)\Bigr)
		\end{align}
	where $\lambda^*_V:=r_V-r$ and $\eta:=\eta(s)$. We now define $I_P$ as the DVA integral in equation \eqref{eq:valuation-forward}:
		\begin{align}
			\label{eq12:proof-fwd-analytical-price}
			I_P:=e^{-r(T-t)}\int_0^{T-t}
			e^{-(r_V-r)\omega}
			\mathcal{BS}_P\left(
			\widetilde{F}_Te^{\kappa(h_S-r_V)\omega},
			\omega\right)\diff\omega
		\end{align}
	Using the formula for the Black-Scholes put function from \eqref{eq:bs-put-formula}:
		\begin{multline}
			\label{eq13:proof-fwd-analytical-price}
			I_P=e^{-r(T-t)}\Biggl(
			K\int_0^{T-t}e^{-(r_V-r)\omega}
			\Phi\left(-\frac{\eta(s)}{\sqrt{\omega}}
			-\left(
			\frac{2\kappa(h_S-r_V)-\sigma^2}{2\sigma}
			\right)\sqrt{\omega}\right)\diff\omega\\
			-\widetilde{F}_T
			\int_0^{T-t}
			e^{-(r_V-r-\kappa(h_S-r_V))\omega}
			\Phi\left(-\frac{\eta(s)}{\sqrt{\omega}}
			-\left(
			\frac{2\kappa(h_S-r_V)+\sigma^2}{2\sigma}
			\right)\sqrt{\omega}\right)\diff\omega
			\Biggr)
		\end{multline}
	Given that:
		\begin{align}
			\label{eq14:proof-fwd-analytical-price}
			\left(-\frac{2\kappa(h_S-r_V)-\sigma^2}
			{2\sigma}\right)^2
			=(-\zeta_2)^2
			=\zeta_2^2,\qquad
			\left(-\frac{2\kappa(h_S-r_V)+\sigma^2}
			{2\sigma}\right)^2
			=(-\zeta_1)^2
			=\zeta_1^2
		\end{align}
	The conditions \eqref{eq10:proof-fwd-analytical-price} remain unchanged for the DVA integral, thus we can express $I_P$ using the function $\Lambda$:
		\begin{align}
			\label{eq17:proof-fwd-analytical-price}
			{\hspace{-.2cm} I_P}=
			e^{-r(T-t)}\left(
			K\Lambda\left(T-t,\lambda^*_V,
			-\zeta_2,-\eta\right)
			-\widetilde{F}_T
			\Lambda\left(T-t,(\lambda^*_V-\kappa(h_S-r_V)),
			-\zeta_1,-\eta\right)
			\right)
		\end{align}
	which concludes the proof.
\end{proof}



\end{document}

%% file: packages/packages_20240324.tex
\usepackage{lmodern} 
\usepackage[T1]{fontenc} 
\usepackage{bbm} 
\usepackage{bm} 

\usepackage[top=1in, bottom=1in, left=1in, right=1in]{geometry}
\usepackage{graphicx}
\usepackage[colorlinks=true,allcolors=blue]{hyperref}
\usepackage[authoryear]{natbib}
\usepackage{titlesec} 
\usepackage[dvipsnames]{xcolor}
\usepackage[normalem]{ulem} 
\usepackage{soul} 
\usepackage{color}
\usepackage{booktabs}
\usepackage[textsize=footnotesize,tickmarkheight=3pt]{todonotes}
\usepackage{setspace}
\usepackage{multirow}
\usepackage{colortbl} 
\usepackage{tabularx}
\usepackage{subcaption}
\usepackage{textgreek} 
\usepackage{float}
\usepackage{threeparttable}

\usepackage{amsmath} 
\usepackage{amsthm} 
\usepackage{amssymb} 
\usepackage{bigints} 
\usepackage{mathrsfs} 
\usepackage{xfrac} 
\usepackage{yfonts} 
\usepackage{mathtools} 
\usepackage{upgreek} 
\usepackage{nccmath} 
\usepackage{scalerel} 
\usepackage{accents} 
\usepackage{nicefrac} 
\usepackage{mathtools} 

\usepackage{subcaption}
\usepackage[labelfont=bf]{caption}

\usepackage{changepage} 
\usepackage{enumitem} 
\usepackage{eurosym} 
\usepackage{indentfirst} 
\usepackage{bibentry} 
\usepackage{pifont} 

\setcitestyle{authoryear,round,semicolon,}

\titleformat{\paragraph}[hang]
{\normalfont\normalsize\bfseries}{}{1em}{}
\titlespacing*{\paragraph}{0pt}{3.25ex plus 1ex minus .2ex}{1.5ex plus .2ex}
\titleformat{\subparagraph}
{\normalfont\normalsize\bfseries}{\thesubparagraph}{1em}{}
\titlespacing*{\subparagraph}{\parindent}{3.25ex plus 1ex minus .2ex}{.75ex plus .1ex}

\makeatletter
\def\today{%
	\two@digits{\the\day}-%
	\ifcase\month\or%
	Jan\or Feb\or Mar\or Apr\or May\or Jun\or%
	Jul\or Aug\or Sep\or Oct\or Nov\or Dec\fi-%
	\number\year%
}
\makeatother

\newtheorem{proposition}{Proposition}
\newtheorem{lemma}{Lemma}
\newtheorem{theorem}{Theorem}
\newtheorem*{theorem*}{Theorem}
\newtheorem{corollary}{Corollary}
\newtheorem{problem}{Problem}
\newtheorem*{problem*}{Problem}
\newtheorem{example}{Example}
\newtheorem{remark}{Remark}

\newcommand{\COMMENTOUT}[1]{} 

\newcommand{\cadlag}{\emph{c\'adl\'ag}}

\newcommand{\diff}{\mathrm{d}}
\newcommand{\indicator}[1]{\mathbbm{1}_{\{#1\}}}
\newcommand{\indicatorf}[2]{\mathbbm{1}_{#1}(#2)}

\newcommand{\qqquad}{\quad\quad\quad}

\DeclareMathOperator{\sgn}{sgn}

\newcommand\bespoketilde[3]{\hstretch{#1}{\tilde{\hstretch{#2}{#3}}}} 
\newcommand\wthalf[1]{\bespoketilde{2}{0.5}{#1}}

\usepackage{ragged2e}
\newcolumntype{C}[1]{>{\hspace{0pt}\Centering\arraybackslash}p{\dimexpr#1\textwidth-2\tabcolsep\relax}}
\newcolumntype{L}[1]{>{\hspace{0pt}\RaggedRight\arraybackslash}p{\dimexpr#1\textwidth-2\tabcolsep\relax}}
\newcolumntype{R}[1]{>{\hspace{0pt}\RaggedLight\arraybackslash}p{\dimexpr#1\textwidth-2\tabcolsep\relax}}

\makeatletter
\newcommand{\vast}{\bBigg@{4}}
\newcommand{\Vast}{\bBigg@{5}}
\makeatother

\newcommand{\checkEmpty}[1]{%
	\if\relax\detokenize{#1}\relax
		{}%
	\else
		{, #1}%
	\fi
}

\newcommand{\mycitepalias}[3]{%
	\citepalias[#1][\citeyear{#3}\checkEmpty{#2}]{#3}%
}
\defcitealias{BCBS2011}{BCBS}
\defcitealias{BCBS2020}{BCBS}
\defcitealias{CICF2012}{CICF}
\defcitealias{EMMI2019}{EMMI}
\defcitealias{FSB2019}{FSB}
\defcitealias{ICEBA2020}{ICEBA}
\defcitealias{ICMA2017}{ICMA}
\defcitealias{ICMA2019}{ICMA}
\defcitealias{ICMA2020}{ICMA}
\defcitealias{ISDA2009}{ISDA}
\defcitealias{ISDA2011}{ISDA}
\defcitealias{ISDA2019a}{ISDA}
\defcitealias{ISDA2019b}{ISDA}
\defcitealias{ISDA2020}{ISDA}
\defcitealias{ISDA2021}{ISDA}
\defcitealias{ISDA2022}{ISDA}
\defcitealias{SEC2017}{SEC}
\defcitealias{ESMA2021}{ESMA}
\defcitealias{ESMA2022}{ESMA}
\defcitealias{BIS2023}{BIS}